\documentclass[aps,prx,reprint,onecolumn,superscriptaddress]{revtex4-2}
\usepackage[T1]{fontenc}
\usepackage[utf8]{inputenc}
\usepackage[a4paper, margin=2cm]{geometry}
\usepackage{physics}
\usepackage{amsthm}
\usepackage{amssymb}
\usepackage{bbold}
\usepackage{mathtools}
\usepackage{graphicx}
\usepackage{booktabs}
\usepackage{multirow}
\usepackage{caption}
\usepackage{subfig}
\usepackage[colorlinks=true, linkcolor=blue, citecolor=green]{hyperref}
\usepackage{natbib}
\usepackage{tikz}

\DeclareMathOperator{\diag}{diag}

\theoremstyle{plain}

\newtheorem{lemma}{Lemma}

\newcommand{\Hp}{\mathcal{H}_{\text{phys}}}

\newcommand{\G}{\mathcal{G}}
\newcommand{\C}{\mathcal{C}}
\newcommand{\m}{\!-\!}
\newcommand{\p}{\!+\!}

\renewcommand{\abs}[2][]{#1|#2#1|}
\renewcommand{\ket}[2][]{#1|#2#1\rangle}

\newcommand{\Ualpha}[2]{X_{(\alpha,#1),(A,#1)}^{(\alpha,#2)}
X_{(\alpha,#1),(B,#1)}^{(\alpha,#2)}
X_{(A,#1),(B,#1)}^{(\alpha,#2)}}

\begin{abstract}
Gauge theories and quantum error-correcting codes share the same underlying structure: both use constraints to identify a specific subspace of the full Hilbert space. In quantum error correction, these constraints are known as stabilizers, while in gauge theories they correspond to Gauss law.
In this work, we consider a family of discrete Abelian lattice gauge theories described by a $\mathbb{Z}_{N}$ gauge group with $N$ an arbitrary power of two.
In this setting, we find a set of stabilizers for the gauge-invariant subspace which is an alternative to the Gauss operators, and we call them binary Gauss stabilizers.
We use this alternative stabilizer group to build practical error-correcting codes exploiting the gauge symmetries of the system without the addition of extra qubits.
The applications of our finding are not limited to error correction though.
We also provide a new strategy of gauge fixing to remove the redundancies based on our alternative stabilizer, which might provide advantages with respect to already-existing approaches such as the axial gauge. Our results provide new tools to study lattice gauge theories and their quantum simulation, and opens directions for future work at the interface of lattice gauge theory and quantum information.
\end{abstract}

\begin{document}

\title{Binary Gauss Stabilizers for Abelian Lattice Gauge Theories}

\newcommand{\UNITN}{{Dipartimento di Fisica, University of Trento, via Sommarive 14, I–38123, Povo, Trento, Italy}}
\newcommand{\TIFPA}{INFN-TIFPA Trento Institute of Fundamental Physics and Applications,  Trento, Italy}

\author{Matteo~Turco}
\email{matteo.turco@tecnico.ulisboa.pt}
\affiliation{Instituto Superior Técnico, Universidade de Lisboa, Lisbon, Portugal}
\affiliation{PQI -- Portuguese Quantum Institute, Lisbon, Portugal}
\affiliation{Physics of Information and Quantum Technologies Group, Centro de Física e Engenharia de Materiais Avançados (CeFEMA), Lisbon, Portugal}
\affiliation{Laboratory of Physics for Materials and Emergent Technologies, Lisbon, Portugal}
\author{Luca~Spagnoli}
\email{luca.spagnoli@unitn.it}
\affiliation{\UNITN}
\affiliation{\TIFPA}
\author{Alessandro~Roggero}
\email{a.roggero@unitn.it}
\affiliation{\UNITN}
\affiliation{\TIFPA}

\maketitle

\section{Introduction}
In the last decades, the growing interest in connections between quantum computation and quantum information on one side and fundamental physics on the other side has been driven by the hope that quantum computers will eventually allow us to simulate from first principles the Standard Model of particle physics in previously inaccessible regimes.
Alongside this extraordinary goal, there have been plenty of works devoted to exploring connections between the two fields on a more theoretical level, or to porting mathematical tools and techniques from one field to the other, in search of new insights~\cite{Araujo:2019mni,Kowalska:2024kbs,Cao:2026mza,Brun:2020ztg,Sellapillay:2020wpf,Sellapillay:2021xai,Liu:2021tef,Siew-Chandrasekharan-Bhattacharya_2026-03}.

As already made manifest by many works, a particularly strong bridge between the two fields is provided by quantum error correction and gauge theories, two key aspects of the two fields.
Gauge theories seem to be unavoidable to describe fundamental models of nature in a $3\p1$-dimensional Minkowski spacetime, and many efforts have been made in order to design quantum algorithms for their digital quantum simulation (for recent reviews see~\cite{Ba_uls_2020, Zohar_2021, Klco_2022, Bauer_PRX_2023,PRXQuantum.5.037001,Bauer-Davoudi-Klco-Martin_2023-06,davoudi2026}).
Quantum error correction plays a crucial role in scaling up computations on quantum platforms due to the presence of noise that, if not suppressed, would void most of the potential advantage promised by quantum computation \cite{ShorQEC, Gottesman1997, Gottesman1998, Knill_qec_98, Terhal2015, Roffe2019}.
The common element of the two topics is the presence of redundancies: in both cases, the Hilbert space of interest is just a subspace of a larger Hilbert space, and it is identified by an extensive number of constraints.
These constraints can be expressed as with the requirement that valid states be simultaneous eigenstates with eigenvalue $+1$ of operators in a given group.
Therefore, both fields make use of stabilizers \cite{Gottesman1997}, although this terminology is not commonly employed in works on gauge theories.
In this field, a common expression is Gauss constraints, and the operators in the stabilizer group are sometimes called Gauss operators.
An important contribution of this manuscript is to further elucidate, with concrete examples, that studying gauge theories and error correction within the same framework yields important practical consequences.

Since the first application of stabilizers is quantum error correction, it is natural to ask whether we can use the constraints of a lattice gauge theory to do error detection/correction and indeed this line of questions has led to several results in the recent past~\cite{PhysRevA.99.042301,rajput2023quantum,Spagnoli-Roggero-Wiebe_2024-05,carrozza2024correspondence,Ballini2025symmetry,Bao2025,Spagnoli:2026qni,Lacambra-ChatwinDavies-Honda-Hoehn_2026-04,Rothlin-Ferradini-Chen_2026-04,Yao_2025-11,pato2026tradeoffsgaussslawerror,Bradshaw_2026-03}.
This question is the original motivation for the present work, and indeed we provide a positive answer for $\mathbb{Z}_N$ lattice gauge theories in one and two spatial dimensions, when $N=2^{\eta}$ is an arbitrary power of two.
In previous works, the relation between error correction and gauge constraints has been studied and understood in the case of $\mathbb{Z}_2$ with or without matter, as well as for $\mathbb{Z}_N$ with $N$ a prime number, and a heavily truncated version of pure-gauge $SU(2)$.
The key point that allowed us to find the result on $\mathbb{Z}_{2^{\eta}}$ is that given a subspace in a larger Hilbert space, there exist many choices of stabilizer groups for it.
This means that in general we are not bound to the stabilizer group that a gauge theory naturally comes with, namely the one generated by the Gauss operators.
Apart from a few known cases in which the Gauss constraints naturally take a particularly suitable form, such as in $\mathbb{Z}_2$ gauge theories, practical use of the Gauss operators to do error correction is far from immediate.
The main contribution of this work is to provide a new set of stabilizers for the gauge-invariant subspace of the theories we considered.
We further show how it is possible to exploit this new stabilizer group to build error correcting codes without the need to add extra qubits.
Therefore, this work also introduces a new viable strategy to develop practical error correcting codes that exploit the symmetries of the physical system.

Error correction is not the only application of our new stabilizers; we also discuss gauge fixing in detail.
Gauge fixing is the process of removing redundancies in a gauge theory to leave only the physical degrees of freedom, and the choice of gauge fixing can significantly impact the practical implementations of the model.
In some cases it is used to reduce the number of qubits required in a quantum simulation~\cite{Farrell-Chernyshev-Powell-Zemlevskiy-Illa-Savage_22-07,Farrell-Illa-Ciavarella-Savage_23-08}.
In other cases it is used as a means to truncate the Hilbert space preserving gauge invariance~\cite{DAndrea:2023qnr,Grabowska:2024emw}.
Here, we present a way to fix the gauge based on the new stabilizers, with the hope that it may offer advantages in existing applications.
In the conclusions, we discuss other possible future directions based on our work.

As already mentioned, the theories we consider in this work have gauge group $\mathbb{Z}_{2^{\eta}}$, and include dynamical staggered fermionic matter on the sites, even though we focus only on the zero-charge sector, and take periodic boundary conditions on hypercubic spatial lattices.
We start by presenting our results in detail in one dimension, which is the simplest case and more convenient for introducing the ideas.
We then proceed to extend most of the results to the two-dimensional case, although in some parts we keep the discussion more concise.
We expect our results to generalize straightforwardly to higher dimensions, different boundary conditions or nonzero-charge sectors, even though we leave explicit verification to future work.

The organization of the manuscript is graphically represented in Figure~\ref{fig:organization_paper}.
In Section~\ref{sec:LGT} we briefly review essential features of $\mathbb{Z}_N$ lattice gauge theories, including mapping to qubits, to fix conventions and notation.
In Section~\ref{sec:binary-stabilizer} we present the main novelty of this manuscript: our alternative stabilizers for the zero-charge sector of the gauge theories.
This Section is split into Subsection~\ref{subsec:binary-stabilizer-1d}, dedicated to the one-dimensional case, and Subsection~\ref{subsec:binary-stabilizer-2d}, dedicated to the two-dimensional case.
The complete derivation of the alternative stabilizer group in two dimensions is detailed in Appendix~\ref{app:binary-stabilizer-2d-generic-eta}.
The following two Sections are dedicated to the applications of our stabilizer group, quantum error correction in Section~\ref{sec:QEC}, and gauge fixing in Section~\ref{sec:gauge-fixing}.
The Section on quantum error correction is divided into three Subsections, dedicated respectively to $\mathbb{Z}_4$ and $\mathbb{Z}_{2^\eta}$ with $\eta>2$ in one dimension, and the  two-dimensional case for $\mathbb{Z}_4$.
In Section~\ref{sec:conclusions} we give our conclusions and perspectives.
The Appendixes~\ref{app:number_physical_states} and~\ref{app:lemma1} further support the main text with some technicalities.

\begin{figure}[t]
    \centering
    \includegraphics[width=.5\linewidth]{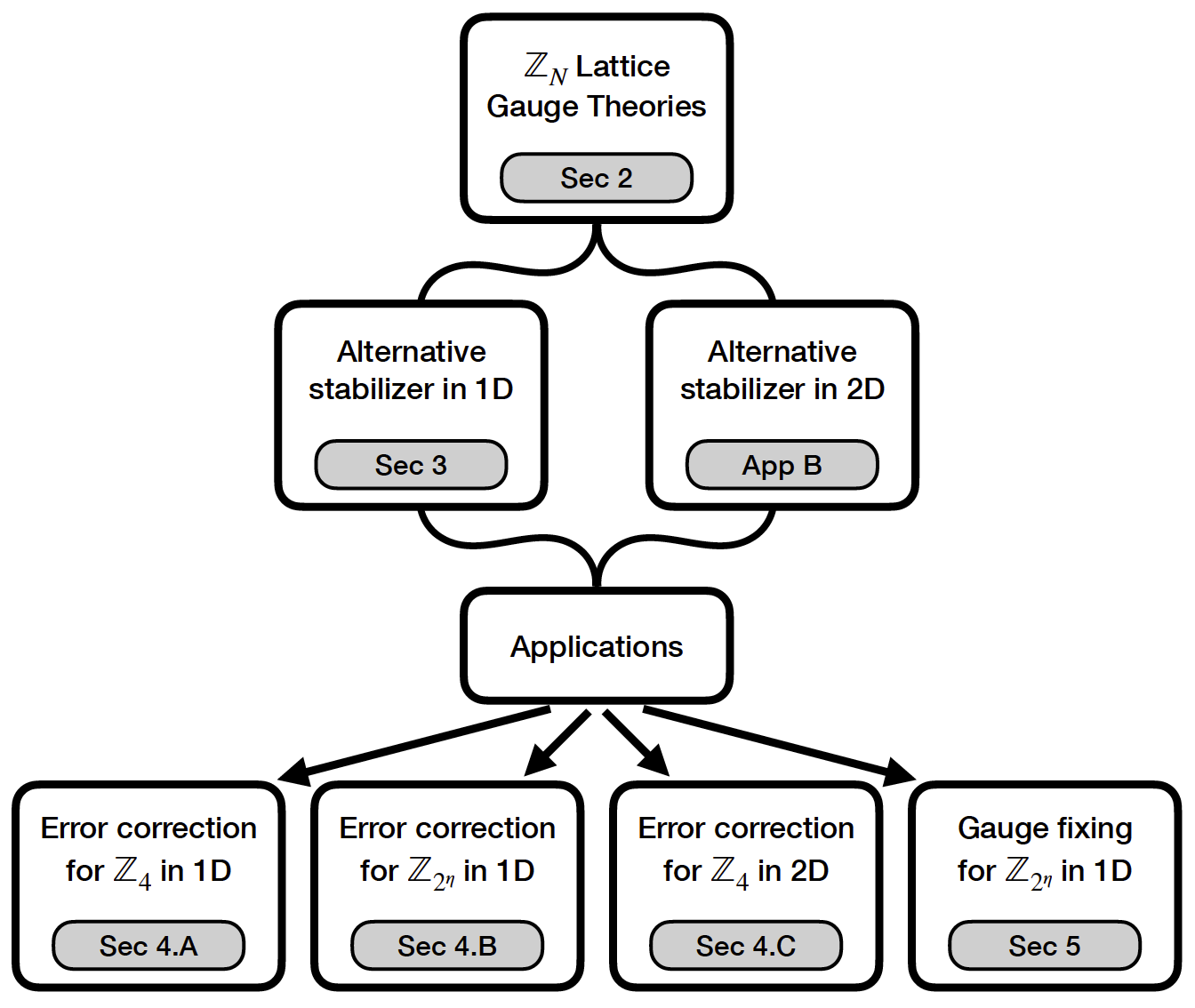}
    \caption{Organization of the paper.}
    \label{fig:organization_paper}
\end{figure}

\section{Notions of $\mathbb{Z}_N$ lattice gauge theories}
\label{sec:LGT}
We consider a gauge theory with Abelian group $\G=\mathbb{Z}_N$ on a hypercubic lattice in $d$ dimensions with fermionic matter and periodic boundary conditions.
The number of sites is $V=L^d$, and the number of links is $E=dV$.
As our purpose is to map this theory onto a qubit system, we take $N=2^{\eta}$ for some integer $\eta\ge1$.
Each site is labeled by $d$ integers $\mathbf{s}=(s_1,\dots,s_d)$, and each link is identified by a pair $(\mathbf{s},l)$, with $l=1,\dots,d$.
We denote the unit vector in the positive direction $l$ as $\hat{\mathbf{l}}$ so that $\hat{\mathbf{1}}=(1,0,0,\dots)$, $\hat{\mathbf{2}}=(0,1,0,\dots)$, and so on.

We place a fermionic mode on each site, $\psi_{\mathbf{s}}$, satisfying the canonical anticommutation relations
\begin{equation}
\{\psi_{\mathbf{s}},\psi_{\mathbf{t}}^{\dagger}\}=\delta_{\mathbf{s}\mathbf{t}},
\qquad
\{\psi_{\mathbf{s}},\psi_{\mathbf{t}}\}=0.
\end{equation}
As local basis, we choose the eigenbasis $\ket{f}_{\mathbf{s}}$ of the number operator
$\psi_{\mathbf{s}}^{\dagger}\psi_{\mathbf{s}}$,
\begin{equation}
\psi_{\mathbf{s}}^{\dagger}\psi_{\mathbf{s}}\ket{f}_{\mathbf{s}}=f\ket{f}_{\mathbf{s}},
\quad f=0,1.
\end{equation}
On every link, we consider the local basis to be the so called electric basis $\ket{m}_{\mathbf{s},l}$.
This is the basis of the irreps of $\mathbb{Z}_N$, which we label by an integer $m=0,\dots,N\m1$.
In this basis, we can define the action of the gauge operators as
\begin{equation}
P_{\mathbf{s},l}\ket{m}_{\mathbf{s},l}=e^{i\frac{2\pi}{N}m}\ket{m}_{\mathbf{s},l},
\qquad
Q_{\mathbf{s},l}\ket{m}_{\mathbf{s},l}=\ket{(m\m1) \mod N}_{\mathbf{s},l}.
\end{equation}
These operators satisfy the $\mathbb{Z}_N$ algebra,
\begin{equation}
P_{\mathbf{s},l}^N=Q_{\mathbf{s},l}^N=1,
\qquad
P_{\mathbf{s},l}^{\dagger}Q_{\mathbf{s},l}P_{\mathbf{s},l}=e^{i\frac{2\pi}{N}}Q_{\mathbf{s},l}.
\end{equation}

The total Hilbert space $\mathcal{H}$ of the system is the tensor product of all the site and link Hilbert spaces.
The physical Hilbert space $\Hp$ is only the subspace identified by the Gauss constraints.
We can define the Gauss constraints using the terminology of the stabilizer formalism.
Specifically, we say that the physical Hilbert space is the subspace stabilized by the Gauss operators
\begin{equation}
G_{\mathbf{s}}=q_{\mathbf{s}}^{\dagger}
\prod_{l=1}^dP_{\mathbf{s},l}^{\dagger}P_{\mathbf{s}-\hat{\mathbf{l}},l},
\label{eq:Gauss_operators}
\end{equation}
where, adopting the staggered-fermion formulation,
\begin{equation}
q_{\mathbf{s}}=e^{-i\frac{2\pi}{N}n_{\mathbf{s}}},
\qquad
n_{\mathbf{s}}=\psi_{\mathbf{s}}^{\dagger}\psi_{\mathbf{s}}
-\frac{1}{2}\big[1-(-1)^{\abs{\mathbf{s}}}\big],
\qquad
\abs{\mathbf{s}}=\sum_{l=1}^d\abs{s_l}.
\end{equation}
More explicitly, this means
\begin{equation}
\Hp=
\big\{\ket{\psi}\in\mathcal{H}\,:\,
G_{\mathbf{s}}\ket{\psi}=\ket{\psi},\,\forall\mathbf{s}\big\}.
\label{zero-charge-sector}
\end{equation}
With this definition, we implicitly choose the zero-charge sector.

For later use, we calculate the number $D$ of states in $\Hp$ using the master formula derived in~\cite{Mariani_2025-09}.
We rewrite it here for square lattices with staggered fermions:
\begin{equation}
D=\sum_{\C}\left(\frac{\abs{\G}}{\abs{\C}}\right)^{E-V}\left[\chi_{\text{even}}(\C)\chi_{\text{odd}}(\C)\right]^{V/2},
\end{equation}
where the sum is over the conjugacy classes $\C$ of $\G$, and $\chi_{\text{even/odd}}(\C)$ is the character of the representation of $\C$ on even/odd sites.
Being $\mathbb{Z}_N$ an Abelian group, each element is a conjugacy class, $\abs{\G}=N$, and $\abs{\C}=1$.
Labeling the conjugacy classes by $c=0,\dots,N-1$, we have
\begin{equation}
\begin{aligned}
\chi_{\text{even}}(c)&=\tr(q_{\mathbf{s}}^c)=1+e^{-i\frac{2\pi c}{N}}
&\qquad&\text{if $\abs{\mathbf{s}}$ even,}\\
\chi_{\text{odd}}(c)&=\tr(q_{\mathbf{s}}^c)=1+e^{i\frac{2\pi c}{N}}
&\qquad&\text{if $\abs{\mathbf{s}}$ odd,}
\end{aligned}
\end{equation}
and therefore we obtain
\begin{equation}
D=N^{E-V}\sum_{c=0}^{N-1}\left[
2+2\cos\left(\frac{2\pi c}{N}\right)
\right]^{V/2}.
\label{number-of-physical-states}
\end{equation}
This expression can be simplified into:
\begin{equation}
    D=N^{E-V+1}\sum_{r\in \mathbb{Z}} \binom{V}{V/2 + rN}\;.
\end{equation}
We provide the full derivation of this expression in Appendix~\ref{app:number_physical_states}. If $N=2$, this formula gives $D=2^{E}$. For $N>2$, there seems to be no case in which $D$ is a power of two, while we provide in the Appendix a formal proof that $D$ is not a power of two if $N>V/2$.
This will be important in the next sections, since $D$ being a power of $2$ is closely related to the physical Hilbert space being stabilized by Pauli operators.

At this point we can introduce the Hamiltonian as follows
\begin{align}
\label{eq:hamiltonian}
H&=H_{\text{mat}}+H_{\text{hop}}+H_{\text{el}}+H_{\text{pl}},\\
H_{\text{mat}}&=m_0\sum_{\mathbf{s}}(-1)^{\abs{\mathbf{s}}}\psi_{\mathbf{s}}^{\dagger}\psi_{\mathbf{s}},\\
H_{\text{hop}}&=\sum_{\mathbf{s}}\sum_{l=1}^d
\left(
\psi_{\mathbf{s}+\hat{\mathbf{l}}}^{\dagger}Q_{\mathbf{s},l}\psi_{\mathbf{s}}+
\psi_{\mathbf{s}}^{\dagger}Q_{\mathbf{s},l}^{\dagger}\psi_{\mathbf{s}+\hat{\mathbf{l}}}
\right),\\
H_{\text{el}}&=g_0\sum_{\mathbf{s}}\sum_{l=1}^d
\left(P_{\mathbf{s},l}+P_{\mathbf{s},l}^{\dagger}\right),\\
H_{\text{pl}}&=\lambda_{\text{pl}}\sum_p
\left(W_p+W_p^{\dagger}\right),
\end{align}
where the sum over $p$ in the last line runs over all the plaquettes.
A plaquette is identified by a site and two directions, $p=(\mathbf{s},k,l)$ with $l>k$, and
\begin{equation}
W_p=
Q_{\mathbf{s},k}
Q_{\mathbf{s}+\hat{\mathbf{k}},l}
Q_{\mathbf{s}+\hat{\mathbf{l}},k}^{\dagger}
Q_{\mathbf{s},l}^{\dagger}.
\end{equation}

All of this is mapped onto a qubit system in the following way.
In this work we employ a Jordan-Wigner representation for the fermions such that the computational basis corresponds with the eigenbasis of $\psi_{\mathbf{s}}^{\dagger}\psi_{\mathbf{s}}$,
\begin{equation}
\psi_{\mathbf{s}}=\zeta(\mathbf{s})\sigma_{\mathbf{s}}^+,\qquad
\psi_{\mathbf{s}}^{\dagger}=\zeta(\mathbf{s})\sigma_{\mathbf{s}}^-,
\end{equation}
where $\zeta(\mathbf{s})$ is a product of $Z$ operators along the arbitrarily chosen path for the representation, and
\begin{equation}
\sigma^{\pm}=\frac{\mathbb{1}\pm Z}{2}X=\mathbb{P}^{\pm}X=X\mathbb{P}^{\mp}.
\end{equation}
In addition, we use $\eta$ qubits per link. The eigenbasis
$\ket{m}$ is mapped to the computational basis. The mapping
is defined by writing the nonnegative
integer $m$ in unsigned binary representation.
With the help of Figure~\ref{fig:intro:levels}, it is useful to divide the qubits into $\eta$ levels across the lattice:
the zeroth level is defined to be the set of qubits on the sites together with the least significant qubits on the links.
The $k$\textsuperscript{th} level, with $k=1,\dots,\eta\m1$, is defined to be the set of qubits associated with the $k$\textsuperscript{th} binary digits of the links.

\begin{figure}
\begin{minipage}{0.5\textwidth}
\includegraphics{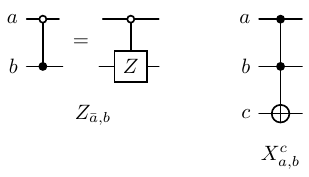}
\caption{
Two examples of the notation introduced in the text, a controlled-$Z$ and a Toffoli operator.
}
\label{notation-examples}
\end{minipage}%
\begin{minipage}{0.5\textwidth}
\includegraphics{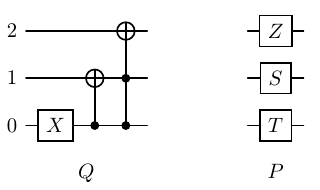}
\caption{
Qubit decompositions for the $Q$ and $P$ operators on a link of the $\mathbb{Z}_8$ theory.
}
\label{link-operators}
\end{minipage}
\end{figure}

\begin{figure}[b]
\includegraphics{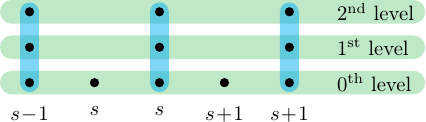}
\caption{
Representation of links and sites encoded on qubits for the $\mathbb{Z}_8$ theory in one dimension, requiring three qubits per link.
The black dots are qubits, the vertical cyan stripes are links.
We also depict with the horizontal green stripes the partition of the qubits in levels as introduced in the text.
}
\label{fig:intro:levels}
\end{figure}

To write a decomposition of the link operators in terms of qubit operators, we introduce now a convenient notation used throughout the manuscript.
We denote by $X_{a_1,\dots,a_k}^b$ the multicontrolled Toffoli gate with controls on the qubits $a_1,\dots,a_k$ and target on $b$, and by $Z_{a_1,\dots,a_k}$ the multicontrolled $Z$ operator with controls on $a_1,\dots,a_k$.
An anticontrol is denoted by a bar over the index.
See Figure~\ref{notation-examples} for a couple of examples.
Furthermore, $X_a$, $Z_a$ and $\diag(\cdot,\cdot)_a$ denote respectively an $X$, a $Z$ and a diagonal operator acting on qubit $a$.
With these conventions, the gauge operators acting on a certain link with qubits $0,\dots,\eta-1$ can be expressed as
\begin{equation}
Q=\left(\prod_{k=\eta-1}^1
X_{0,\dots,k-1}^k\right)X_0=
X_{0,\dots,\eta-2}^{\eta-1}\cdots X_0^1X_0\qquad
P=\prod_{k=0}^{\eta-1}\diag(1,e^{i\pi/2^{\eta-k-1}})_{(i,k)}.
\label{eq:intro:Q-and-P}
\end{equation}
In Figure~\ref{link-operators} we show the operators
of $\mathbb{Z}_8$ as an example.

\section{Binary Gauss stabilizers for the physical subspace}
\label{sec:binary-stabilizer}
In Refs.~\cite{rajput2023quantum,Spagnoli-Roggero-Wiebe_2024-05}, it is shown that the qubit-encoded, physical Hilbert space of a $\mathbb{Z}_2$ theory is stabilized by a set of Pauli stabilizers. The construction has been recently extended to $\mathbb{Z}_N$ theories with $N$ prime where a qudit mapping is similarly stabilized by generalized Pauli stabilizers~\cite{Spagnoli:2026qni}.
This is no longer the case when we consider a $\mathbb{Z}_{2^\eta}$ theory with $\eta>1$.
To see why, it is enough to look at the dimension $D$ of $\mathcal{H}_{\text{phys}}$.
Considering Eq.~\eqref{number-of-physical-states}, $D$ is a power of $2$ for $N=2$, while for $N>2$, we conjecture that $D$ is not a power of $2$.
On the other hand, it is easy to show that the dimension of a space stabilized by a Pauli stabilizer group is always a power of $2$.
Thus, it may be possible to find a Pauli stabilizer group only if the dimension of $\mathcal{H}_{\text{phys}}$ is a power of $2$.
We conclude that $\mathcal{H}_{\text{phys}}$ cannot be stabilized by a Pauli stabilizer group for $N>2$.

Rather, what we find here is, in one dimension, a description of $\Hp$ in terms of stabilizers made of products of $Z$ and controlled-$Z$ operators.
In higher dimensions, we need also multicontrolled-$Z$ operators, with the number of controls growing with the number $\eta$ of qubits per link.

\subsection{One dimension}
\label{subsec:binary-stabilizer-1d}

Let us start for simplicity by considering the one-dimensional case.
In one dimension, a site and its right link are both identified by a single integer $s$.
We label the qubit of site $s$ by the same index $s$, and the
$k$\textsuperscript{th}-digit qubit of link $s$ by $(s,k)$, in such a way that
$(s,\eta\m1)$ is the most-significant qubit of link $s$.
Due to the use of staggered fermions, qubits on even links often have anticontrols, while qubits on odd links often have controls.
For this reason, for the link qubits we introduce the following notation,
\begin{equation}
\widetilde{(s,k)}=
\begin{cases}
\overline{(s,k)} & \text{if $s$ is even,}\\[2mm]
(s,k) & \text{if $s$ is odd.}
\end{cases}
\label{eq:1d:tilde-notation}
\end{equation}

Consider the Gauss operators $G_s$ of Eq.~\eqref{eq:Gauss_operators}. Remember that we called the physical Hilbert space of the theory $\Hp$ the space stabilized by $G_s$, for every $s$. This means that $\Hp$ is spanned by states that are simultaneous eigenvectors with $+1$ eigenvalues of $G_s$ for every $s$.
The aim of this section is to introduce an alternative set of operators, which we denote by $g_{s,k}$, with $k=0,\cdots,\eta-1$, such that the space stabilized by all $g_{s,k}$ is still $\Hp$.

The new operators read
\begin{equation}
\begin{aligned}
    g_{s,0} &= (-1)^sZ_{(s-1,0)}Z_sZ_{(s,0)}\;, \\
    g_{s,k} &= Z_{(s-1,k)}Z_{(s,k)} Z_{\widetilde{(s-1,k-1)},\widetilde{(s,k-1)}}\;,
\end{aligned}
\label{eq:1d_generators}
\end{equation}
and they generate the stabilizer group $\mathcal{S}_g$.
We have
\begin{equation}
\begin{aligned}
\Hp&=
\left\{
\ket{\psi}\in\mathcal{H}\,:\,
G_s\ket{\psi}=\ket{\psi}\quad
\forall s=0,\dots,V-1
\right\}\\
&=
\left\{
\ket{\psi}\in\mathcal{H}\,:\,
g_{s,k}\ket{\psi}=\ket{\psi}\quad
\forall s=0,\dots,V-1\quad
\forall k=0,\dots,\eta-1
\right\}.
\end{aligned}
\label{eq:1d_stabilizer-equivalence}
\end{equation}
The three generators $g_{s,0}$, $g_{s,1}$ and $g_{s,2}$ of the stabilizer group for $\mathbb{Z}_8$ are depicted in Figure~\ref{binary-Gauss-stabilizer-Z8-1d}.
The rest of this subsection is dedicated to the proof of Eq.~\eqref{eq:1d_stabilizer-equivalence}.

\begin{figure}
\includegraphics{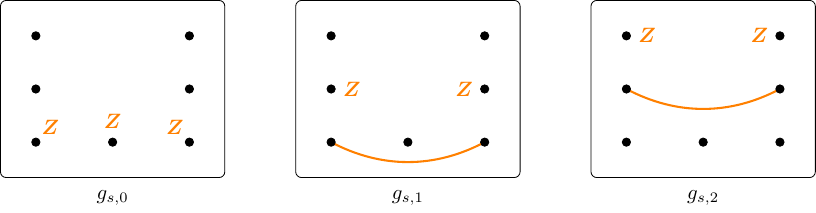}
\caption{
Generators of the binary Gauss stabilizer group for the $\mathbb{Z}_8$ theory in one dimension.
The orange lines represent controlled-$Z$ operators.
In each box, we show the qubits of site $s$ and the two links $s\m1$ and $s$.
}
\label{binary-Gauss-stabilizer-Z8-1d}
\end{figure}

We first prove that the set of states stabilized by $G_s$ coincides with the set of states stabilized by
\begin{equation}
\begin{aligned}
g'_{s,0}&=(-1)^sZ_{(s-1,0)}Z_sZ_{(s,0)}\\
g'_{s,k}&=(-1)^sZ_{(s-1,k)}Z_sZ_{(s,k)}\prod_{l=0}^{k-1}
Z_{\widehat{(s-1,l)},\widehat{(s,l)}}\qquad k=1,\dots,\eta-1,
\end{aligned}
\end{equation}
where
\begin{equation}
\widehat{(s,k)}=
\begin{cases}
(s,k) & \text{if $s$ is even,}\\[2mm]
\overline{(s,k)} & \text{if $s$ is odd.}
\end{cases}
\end{equation}
Then, notice that the operators $g'_{s,k}$ and $g_{s,k}$ generate the same stabilizer group, since one set of generators can be written as products of the other generators,
\begin{equation}
g_{s,k}=g'_{s,k-1}g'_{s,k},
\qquad\qquad
g'_{s,k}=\prod_{l=0}^kg_{s,k}.
\end{equation}
These two identities can be checked explicitly using $Z_{\bar{a},b}Z_aZ_b=Z_{a,\bar{b}}$.

Let us focus on a single constraint $G_s$ and consider only the registers of qubits it acts on. Let us take a state
\begin{equation}
\ket{\psi_s}=\ket{a_{\eta-1}\dots a_0}_{s-1}\ket{f_s}_s\ket{b_{\eta-1}\dots b_0}_s\qquad\qquad
m_{s-1}=\sum_{k=0}^{\eta-1}a_k2^k\qquad
m_s=\sum_{k=0}^{\eta-1}b_k2^k.
\label{m-decomposition}
\end{equation}
The constraint $G_s\ket{\psi_s}=\ket{\psi_s}$ is equivalent to
\begin{align}
m_s=(m_{s-1}+f_s)&\mod 2^{\eta}\qquad\text{if $s$ is even,}
\label{eq:1d:even-Gauss-law}\\
m_{s-1}=(m_s+1-f_s)&\mod 2^{\eta}\qquad\text{if $s$ is odd.}
\end{align}
Our aim is to write these relations in their bit-wise form.
For simplicity, we can focus on even $s$, since the case of odd $s$ is completely analogous after the substitutions $f_s\rightarrow(1-f_s)$ and $m_s\leftrightarrow m_{s-1}$.
The carry bits of the sum $m_{s-1}+f_s$ are
\begin{equation}
\begin{aligned}
\alpha_1&=a_0f_s\\
\alpha_k&=a_{k-1}\alpha_{k-1}\quad k=2,\dots,\eta-1.
\end{aligned}
\label{eq:1d:carries}
\end{equation}
Then, the number $m_{s-1}+f_s\mod 2^{\eta}$ in binary representation is $\sigma_{\eta-1}\dots\sigma_0$, with
\begin{equation}
\begin{aligned}
\sigma_0&=a_0\oplus f_s\\
\sigma_k&=a_k\oplus \alpha_k\quad k=1,\dots,\eta-1,
\end{aligned}
\end{equation}
and Eq.~\eqref{eq:1d:even-Gauss-law} in bit-wise form is $\sigma_k=b_k$, or
\begin{equation}
\begin{aligned}
a_0\oplus b_0\oplus f_s&=0\\
a_k\oplus b_k\oplus \alpha_k&=0\quad k=1,\dots,\eta-1.
\end{aligned}
\end{equation}
Summarizing, we have found that
\begin{equation}
G_s\ket{\psi_i}=\ket{\psi_s}\iff
\begin{cases}
a_0\oplus b_0\oplus f_s=0\\
a_k\oplus b_k\oplus\alpha_k=0\quad k=1,\dots,\eta-1
\end{cases}.
\label{bit-wise-Gauss}
\end{equation}
As shown in Lemma~\ref{lemma:ck-equivalence} (see Appendix.~\ref{app:lemma1}) with $\alpha_0=f_s$, we can solve the recursion relations in Eq.~\eqref{eq:1d:carries} and write the carry bits more explicitly as
\begin{equation}
\alpha_k=f_s\oplus\left(\bigoplus_{l=0}^{k-1}\bar{a}_lb_l\right)\;.
\end{equation}

On the other hand, $Z_{\overline{(s-1, l)},(s,l)}\ket{\psi_s}=(-1)^{\bar{a}_lb_l}\ket{\psi_s}$, and thus
\begin{equation}
Z_s \prod_{l=0}^{k-1}Z_{\overline{(s-1, l)}, (s,l)}\ket{\psi_s} = (-1)^{f_s} \prod_{l=0}^{k-1} (-1)^{\bar{a}_lb_l}\ket{\psi_s} = (-1)^{f_s \oplus \left(\bigoplus_{l=0}^{k-1}(\bar{a}_lb_l)\right)}\ket{\psi_s}=
(-1)^{\alpha_k}\ket{\psi_s}.
\end{equation}
From this equation, it follows immediately that $g'_{s,k}\ket{\psi_s}=(-1)^{a_k\oplus b_k\oplus \alpha_k}\ket{\psi_s}$, and so
\begin{equation}
g'_{s,k}\ket{\psi_s}=\ket{\psi_s}\quad\forall k=0,\dots,\eta-1
\quad\iff\quad
\begin{cases}
a_0\oplus b_0\oplus f_s=0\\
a_k\oplus b_k\oplus\alpha_k=0\quad k=1,\dots,\eta-1
\end{cases}.
\end{equation}
Comparing with Eq.~\eqref{bit-wise-Gauss}, the proof of Eq.~\eqref{eq:1d_stabilizer-equivalence} is complete.

In summary, our main result is to replace each Gauss operator by $\eta$ operators encoding the Gauss' law in binary form.
For this reason we call them binary Gauss stabilizers.
This is an alternative way to describe the physical subspace of a gauge theory, and it may be used as a new tool to analyze such theories.
In Sections~\ref{sec:QEC} and~\ref{sec:gauge-fixing}, we prove its usefulness by showing two important applications, one in error correction and the other in removing the gauge redundancies.

\subsection{Two dimensions}
\label{subsec:binary-stabilizer-2d}
We now show how to generalize the stabilizer group of the previous section to higher dimensions, focusing on the two-dimensional case.
The idea described here should be straightforwardly applicable to arbitrary dimensions.
To keep the discussion concise, we show only the case of $\mathbb{Z}_4$ in the main text, and the general case of $\mathbb{Z}_{2^{\eta}}$ in Appendix~\ref{app:binary-stabilizer-2d-generic-eta}.

Given a site $\mathbf{s}$, we denote by $A$, $B$, $C$ and $D$ the east, north, west and south links attached to $\mathbf{s}$.
If we take
\begin{equation}
\label{eq:vertex_state}
\ket{\psi_{\mathbf{s}}}=\ket{m_A}_A\ket{m_B}_B
\ket{f_\mathbf{s}}_{\mathbf{s}}
\ket{m_C}_C\ket{m_D}_D=
\ket{a_{\eta-1}\dots a_0}_A\ket{b_{\eta-1}\dots b_0}_B
\ket{f_{\mathbf{s}}}_{\mathbf{s}}
\ket{c_{\eta-1}\dots c_0}_C\ket{d_{\eta-1}\dots d_0}_D\;,
\end{equation}
the constraint $G_{\mathbf{s}}\ket{\psi_{\mathbf{s}}}=\ket{\psi_{\mathbf{s}}}$ is equivalent to
\begin{align}
(m_A+m_B)\mod 2^{\eta}&=(m_C+m_D+f_{\mathbf{s}})
\mod 2^{\eta}&\quad&\text{if $\abs{\mathbf{s}}$ is even,}
\label{eq:2d:even-conservation-law}\\
(m_C+m_D)\mod 2^{\eta}&=(m_A+m_B+1-f_{\mathbf{s}})
\mod 2^{\eta}&\quad&\text{if $\abs{\mathbf{s}}$ is odd.}
\label{eq:2d:odd-conservation-law}
\end{align}

For simplicity, let us focus on a site with even $\abs{\mathbf{s}}$, as the odd case follows immediately after minor changes.
We start by writing the bit-wise version of the relation in Eq.~\eqref{eq:2d:even-conservation-law}.
To do this, we first need to compute the carry bits $\alpha_k$ and $\beta_k$ of the additions $m_A+m_B$ and $m_C+m_D+f_{\mathbf{s}}$.
Setting $\alpha_0=0$ and $\beta_0=f_{\mathbf{s}}$, we have the recursion relations
\begin{equation}
\alpha_k=\alpha_{k-1}a_{k-1}
\oplus\alpha_{k-1}b_{k-1}
\oplus a_{k-1}b_{k-1}\qquad
\beta_k=\beta_{k-1}c_{k-1}
\oplus\beta_{k-1}d_{k-1}
\oplus c_{k-1}d_{k-1}
\label{eq:2d:carries}
\end{equation}
for $k=1,\dots,\eta-1$.
Then, the numbers $m_A+m_B\,\mod2^{\eta}$ and $m_C+m_D+f_{\mathbf{s}}\,\mod2^{\eta}$ in binary representation are respectively $\sigma_{\eta-1}\dots\sigma_0$ and $\tau_{\eta-1}\dots\tau_0$, with
\begin{equation}
\sigma_k=a_k\oplus b_k\oplus\alpha_k,
\qquad
\tau_k=c_k\oplus d_k\oplus\beta_k,
\qquad k=0,\dots,\eta-1.
\end{equation}
The bit-wise version of Eq.~\eqref{eq:2d:even-conservation-law} is $\sigma_k=\tau_k$, or
\begin{align}
&a_0\oplus b_0\oplus c_0\oplus d_0\oplus f_{\mathbf{s}}=0
\label{eq:2d:bit-wise-conservation-laws-0}\\
&a_k\oplus b_k\oplus c_k\oplus d_k\oplus\alpha_k\oplus\beta_k=0
\qquad
k=1,\dots,\eta-1.
\label{eq:2d:bit-wise-conservation-laws-k}
\end{align}

To find the stabilizer group, we proceed in two steps.
First, we find circuits checking whether the relations in Eq.~\eqref{eq:2d:bit-wise-conservation-laws-0} and~\eqref{eq:2d:bit-wise-conservation-laws-k} hold.
These circuits require two ancillary registers where to compute the carries $\alpha_k$ and $\beta_k$.
Second, we use commutation identities between multicontrolled-$X$ and $Z$ gates to simplify the circuits and remove the ancillary registers.
The circuits are only a formal tool, and do not need to be implemented in practice.

\begin{figure}
\includegraphics{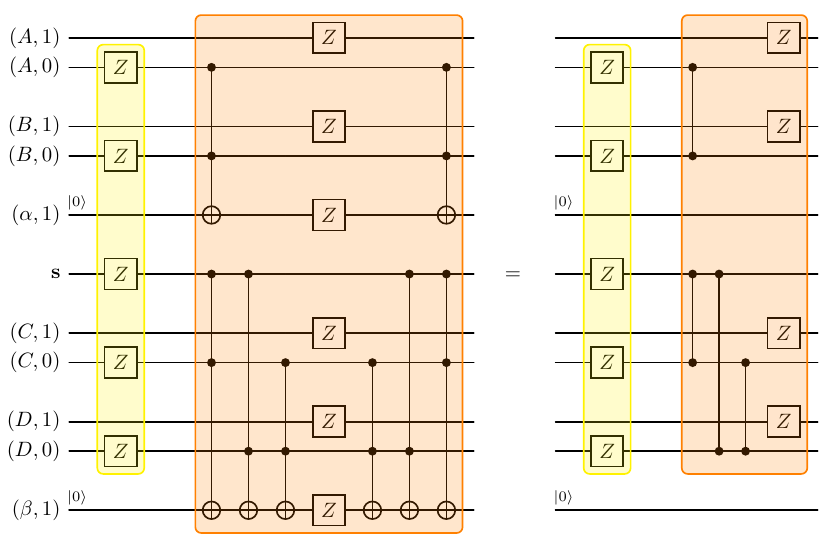}
\caption{
Generators for the stabilizer group of the $\mathbb{Z}_4$ theory in two dimensions.
The ancillary qubits $(\alpha,1)$ and $(\beta,1)$ are initially set in the $\ket{0}$ state.
The operator in the yellow box is $g_{\mathbf{s},0}$, the one in the orange box on the left-hand side is $g'_{\mathbf{s},1}$, and the one in the orange box on the right-hand side is $g_{\mathbf{s},1}$.
}
\label{fig:2d:Z4-stabilizer-circuit}
\end{figure}

A detailed example for $\mathbb{Z}_4$ is shown in Figure~\ref{fig:2d:Z4-stabilizer-circuit}.
The operator in the yellow box is $g_{\mathbf{s},0}$, whose action is
\begin{equation}
g_{\mathbf{s},0}\ket{\psi_{\mathbf{s}}}=
(-1)^{a_0\oplus b_0\oplus c_0\oplus d_0\oplus f_{\mathbf{s}}}
\ket{\psi_{\mathbf{s}}}.
\end{equation}
Thus, we see that the eigenstates with eigenvalue $1$ of this operator are the states that satisfy Eq.~\eqref{eq:2d:bit-wise-conservation-laws-0}.

We denote by $g'_{\mathbf{s},1}$ the operator in the orange box on the left-hand side of Figure~\ref{fig:2d:Z4-stabilizer-circuit}.
In the case of $\mathbb{Z}_4$, we only need one ancillary qubit $(\alpha,1)$ for $\alpha_1=a_0b_0$, and one ancillary qubit $(\beta,1)$ for $\beta_1=f_{\mathbf{s}}c_0\oplus f_{\mathbf{s}}d_0\oplus c_0d_0$.
Both the ancillae are initially set in the $\ket{0}$ state.
Recalling the action of a Toffoli gate on the computational basis,
\begin{equation}
X_{a,b}^c\ket{x}_a\ket{y}_b\ket{z}_c=\ket{x}_a\ket{y}_b\ket{z\oplus xy}_c,
\label{eq:2d:Toffoli-action}
\end{equation}
we compute $\alpha_1$ and $\beta_1$ with one and three Toffoli gates respectively, on the left-most part of Figure~\ref{fig:2d:Z4-stabilizer-circuit}.
After these four Toffoli gates, the six $Z$ operators in the orange box together give
\begin{equation}
(-1)^{a_1\oplus b_1\oplus c_1\oplus d_1\oplus\alpha_1\oplus\beta_1}
\ket{\psi_{\mathbf{s}}}\ket{\alpha_1}\ket{\beta_1}.
\label{eq:2d:six-Z-operators}
\end{equation}
After the $Z$ operators, we need to uncompute $\alpha_1$ and $\beta_1$ with the last four Toffoli gates.
As a consequence of Eq.~\eqref{eq:2d:six-Z-operators}, the six $Z$ operators in the orange box act like the identity on the states satisfying Eq.~\eqref{eq:2d:bit-wise-conservation-laws-k}.
Then, if the $Z$ operators act like the idenity, the eight Toffoli gates cancel against each other.
Therefore, we see that the eigenstates with eigenvalue $1$ of $g'_{\mathbf{s},1}$ are the states that satisfy Eq.~\eqref{eq:2d:bit-wise-conservation-laws-k}.
This completes the first step of our procedure to find $g_{\mathbf{s},1}$.

We now describe the second step, which is needed to simplify this operator, and obtain the final form of $g_{\mathbf{s},1}$, shown on the right-hand side of Figure~\ref{fig:2d:Z4-stabilizer-circuit}.
Starting from the left-hand side, and using the identity $X_{a,b}^cZ_cX_{a,b}^c=Z_cZ_{a,b}$ shown in Figure~\ref{fig:2d:conjugation-identity}, we move the $Z$ operators on the ancillary qubits to the left of the circuit.
Finally, we notice that $Z\ket{0}=\ket{0}$, so we can remove the two $Z$ from the ancillary qubits. What remains is the operator $g_{\mathbf{s},1}$ in the orange box on the right-hand side of Figure~\ref{fig:2d:Z4-stabilizer-circuit}.
We have that $g_{\mathbf{s},1}$ and $g'_{\mathbf{s},1}$ are the same on the subspace in which the ancillary qubits are in the $\ket{0}$ state.

We can proceed in the same way to find the generators of $\mathbb{Z}_{2^{\eta}}$ with $\eta>2$, as detailed in Appendix~\ref{app:binary-stabilizer-2d-generic-eta}.
We will find again the same $g_{\mathbf{s},0}$ and $g_{\mathbf{s},1}$.
In addition, we will find $g_{\mathbf{s},k}$ with $k$ up to $\eta-1$, one for each relation in Eq.~\eqref{eq:2d:bit-wise-conservation-laws-k}.
All the generators share the common structure
\begin{equation}
g_{\mathbf{s},k}=
Z_{(A,k)}\otimes
Z_{(B,k)}\otimes
Z_{(C,k)}\otimes
Z_{(D,k)}\otimes
\mathcal{C}_{\mathbf{s},k},
\label{eq:2d:generator-structure}
\end{equation}
where $\mathcal{C}_{\mathbf{s},k}$ is a product of controlled-$Z$ and multicontrolled-$Z$ operators acting on site $\mathbf{s}$ and on the links $A$, $B$, $C$ and $D$ up to the level $k\m1$.
The number of operators in $\mathcal{C}_{\mathbf{s},k}$ grows exponentially with $k$, and it contains up to $(k\p1)$-controlled-$Z$ operators.

\begin{figure}[b]
\includegraphics{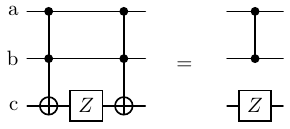}
\caption{Circuital representation of the identity $X_{a,b}^cZ_cX_{a,b}^c=Z_cZ_{a,b}$.}
\label{fig:2d:conjugation-identity}
\end{figure}

We conclude this section as we opened it, with a remark on the number of states in the physical space.
Like the operators in the Pauli group, the binary Gauss stabilizers are reflections, in the sense that they only have $\pm 1$ eigenvalues and they square to the identity.
Unlike the operators in the Pauli group, not all of them are traceless.
All the $g_{\mathbf{s},k}$ generators we wrote are traceless, as they are tensor products of two or more Pauli operators with other operators, and $\tr(A\otimes B)=\tr(A)\tr(B)$.
However, the operators
\begin{equation}
g_k=\prod_{\mathbf{s}}g_{\mathbf{s},k},\qquad
k=1,\dots,\eta-1,
\end{equation}
no longer contain Pauli operators, as the $Z$ operators all cancel against each other.
It can also be checked that the operators $g_k$ are not traceless.
Therefore, they are responsible for the fact that the number of physical states is not a power of two.
This is because the number of $+1$ eigenvalues of $g_k$ is not half the total number of states in $\mathcal{H}$, and so, the number of states selected by $g_k$ is not a power of two.

\section{Application: error correction}
\label{sec:QEC}

So far, we have introduced the binary Gauss operators, denoted by $g_{\mathbf{s},k}$, which stabilize the same physical Hilbert space as the Gauss operators $G_{\mathbf{s}}$.
In this and the next sections, we will provide some applications of these new stabilizers, highlighting their relevance as an independent way to study the gauge constraints of lattice gauge theories.
In particular, we will focus on two main directions:
how to use the constraints to do error correction, generalizing the approach of $\mathbb{Z}_2$, and how to enforce the gauge symmetry, removing the gauge redundancies and the memory overhead to encode the nonphysical part of the Hilbert space.

\subsection{Error correction for $\mathbb{Z}_4$ in one dimension}
\label{subsec:QEC-1d-eta2}
In this subsection we construct a Pauli-stabilizer error-correcting code starting from the stabilizer group introduced in the previous section. We show the construction explicitly in the case of $\mathbb{Z}_4$ in one dimension, but the idea holds for $\mathbb{Z}_{2^{\eta}}$ with larger $\eta$, and in higher dimensions. We will briefly discuss the generalization in the subsequent subsections.

The idea is to apply a unitary transformation $U$ to turn some of the non-Pauli generators $g_{s,k}$ into Pauli ones.
Then, for the purposes of error correction, we can forget about the remaining non-Pauli generators, and use only the Pauli ones.
The price to pay is a constant-factor increase in the cost to implement the operators in the Hamiltonian.
Suppose that $g$ stabilizes a state $\ket{\psi}$, namely $g\ket{\psi}=\ket{\psi}$.
If we apply a change of basis $U$, we get
\begin{equation}
    U\ket{\psi} = Ug\ket{\psi} = UgU^\dagger U\ket{\psi} = \tilde{g} U\ket{\psi}.
\end{equation}
In words, the transformed state $U\ket{\psi}$ is stabilized by the conjugated stabilizer $\tilde{g}=UgU^{\dagger}$.
We will also use the identity
\begin{equation}
X_{a,b}^cZ_cX_{a,b}^c=Z_cZ_{a,b},
\label{eq:QEC:conjugation-identity}
\end{equation}
shown in Figure~\ref{fig:2d:conjugation-identity}. We will choose $U$ to be an appropriate product of Toffoli gates to cancel controlled-$Z$ operators in the generators with the controlled-$Z$ appearing in Eq.~\eqref{eq:QEC:conjugation-identity}.

Consider now the stabilizer group $\mathcal{S}_g$ of $\mathbb{Z}_4$, generated by the operators of Eq.~\eqref{eq:1d_generators},
\begin{equation}
g_{s,0}=(-1)^sZ_{(s-1,0)}Z_sZ_{(s,0)},\qquad
g_{s,1}=Z_{(s-1,1)}Z_{(s,1)}Z_{\widetilde{(s-1,0)}\widetilde{(s,0)}}\;,
\end{equation}
on every site $s$.
Furthermore, define the unitary operator
\begin{equation}
U_{s-1\leftarrow s}=X_{\widetilde{(s-1,0)}\widetilde{(s,0)}}^{(s-1,1)}.
\label{eq:QEC:one-step-left}
\end{equation}
This is a Toffoli operator obtained by attaching a target on qubit $(s\m1,1)$ to the controlled-$Z$ operator $Z_{\widetilde{(s-1,0)}\widetilde{(s,0)}}$.
When conjugating $\mathcal{S}_g$ with this transformation, we get
\begin{equation}
\begin{aligned}
g_{s,1}&\longrightarrow U_{s-1\leftarrow s}\,g_{s,1}\,U_{s-1\leftarrow s}^{\dagger} = Z_{(s-1,1)}Z_{(s,1)}\\
g_{s-1,1}&\longrightarrow U_{s-1\leftarrow s}\,g_{s-1,1}\,U_{s-1\leftarrow s}^{\dagger} = Z_{(s-2,1)}Z_{(s-1,1)} Z_{\widetilde{(s-2,0)}\widetilde{(s-1,0)}} Z_{\widetilde{(s-1,0)}\widetilde{(s,0)}}.
\end{aligned}
\end{equation}
We have turned the non-Pauli generator $g_{s,1}$ into a Pauli operator, while $g_{s-1,1}$ after the transformation carries two controlled-$Z$ operators.
In short, we can say that we have moved the controlled-$Z$ operator from $g_{s,1}$ to the transform of $g_{s-1,1}$.
All the other generators of the stabilizer group are left untouched because $U_{s-1\leftarrow s}$ commutes with them.
Similarly, we can define
\begin{equation}
U_{s\rightarrow s+1}=X_{\widetilde{(s-1,0)}\widetilde{(s,0)}}^{(s,1)},
\label{eq:QEC:one-step-right}
\end{equation}
which moves the controlled-$Z$ operator from $g_{s,1}$ to the transform of $g_{s+1,1}$.

For simplicity, let us assume the number of sites to be a (even) multiple of 3, and denote with $\mathcal{A}$ the subset of sites labeled by the multiples of 3, $\mathcal{A}=\{s\,:\,s\mod3=0\}$.
Then, the transformation
\begin{equation}
U=\prod_{a\in\mathcal{A}}U_{a-1\rightarrow a}U_{a\leftarrow a+1}
\label{eq:QEC:change-of-basis}
\end{equation}
acts on the generators of $\mathcal{S}_g$ as
\begin{equation}
\begin{aligned}
\tilde{g}_{s,0}&=g_{s,0}\qquad\forall s,\\[4mm]
\tilde{g}_{s,1}&=
\begin{cases}
Z_{(s-1,1)}Z_{(s,1)}\prod_{t=s-1}^{s+1}
Z_{\widetilde{(t-1,0)},\widetilde{(t,0)}} & \,\,s\in\mathcal{A},\\
Z_{(s-1,1)}Z_{(s,1)} &\,\,\text{otherwise.}
\end{cases}
\end{aligned}
\label{eq:QEC:transformed-generators}
\end{equation}
We see that the generators $g_{s,1}$ with $s\notin\mathcal{A}$ have been turned into Pauli operators, while the generators $g_{s,0}$ have not been changed.
All the controlled-$Z$ operators are now carried by the generators $\tilde{g}_{a,1}$ with $a\in\mathcal{A}$. The transformation is visualized in Figure~\ref{fig:QEC:accumulating-CZ-1d}.

\begin{figure}
\includegraphics{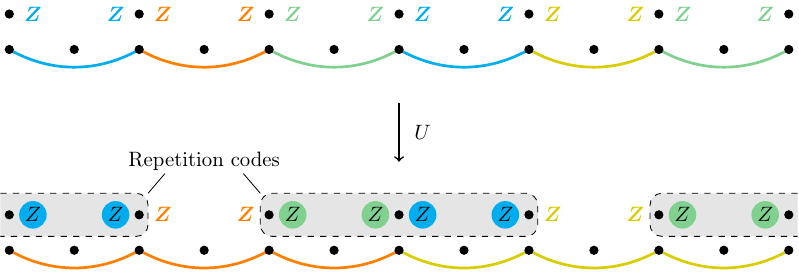}
\caption{
Effect of the transformation $U$ on the first-level generators.
The coloured lines are controlled-$Z$ operators defined according the tilde notation in Eq.~\eqref{eq:1d:tilde-notation}, and the black dots are qubits.
After the transformation, the Pauli generators are highlighted by the coloured circles.
A dashed box highlights a triplet of qubits in the repetition code.
}
\label{fig:QEC:accumulating-CZ-1d}
\end{figure}

For the purposes of error correction, we can forget about the non-Pauli generators, and focus only on the Pauli generators.
In fact, the $\tilde{g}$ Pauli generators are enough to build a distance-three classical code, which can be used to correct bit-flip errors in a standard way.
To see this, first notice that the generators $\tilde{g}_{s,0}=g_{s,0}$ form exactly the $\mathbb{Z}_2$-gauge-covariant code of Reference~\cite{Spagnoli-Roggero-Wiebe_2024-05}, so they are enough to correct single-qubit errors on the site qubits and the link qubits $(s,0)$.
The remaining Pauli generators $\tilde{g}_{s,1}$ with $s\notin\mathcal{A}$ are nothing but copies of the three-qubit repetition code, one for each $a\in\mathcal{A}$, highlighted by the dashed boxes in Fig.~\ref{fig:QEC:accumulating-CZ-1d}.
For a given $a$, the copy is formed by $\tilde{g}_{a+1,1}$ and $\tilde{g}_{a+2,1}$.

The parameters of the code can be found by considering the number of physical qubits, and the number of Pauli stabilizers.
In order to encode the system with $V$ sites and $V$ links we used $2V+V=3V$ physical qubits.
The constraints are $V$ on the zeroth level, $\tilde{g}_{s,0}$, and $2V/3$ on the first level, $\tilde{g}_{s,1}$ with $s\notin\mathcal{A}$, for a total of $5V/3$.
Thus, we have obtained a $[3V, 4V/3, 3]$ classical code using only the gauge symmetry.
The transformed physical space $\tilde{\mathcal{H}}_{\text{phys}}=U\mathcal{H}_{\text{phys}}U^{\dagger}$ is embedded in the codespace, since we have obtained the codespace by removing some of the constraints that identify the physical space.

Since we have applied the change of basis $U$ in Eq.~\eqref{eq:QEC:change-of-basis}, we also need to transform the relevant operators appearing in the $\mathbb{Z}_4$ Hamiltonian of Eq.~\eqref{eq:hamiltonian}.
The fermionic operators $\psi_s, \psi_s^\dagger$ are not modified, as $U$ acts only on the links.
We then need to find the form for the $P_s$ and $Q_s$ operators in Eq.~\eqref{eq:intro:Q-and-P} in the new basis.
For the case of $\mathbb{Z}_4$ Eq.~\eqref{eq:intro:Q-and-P} becomes explicitly
\begin{equation}
Q_s=X_{(s,0)}^{(s,1)}X_{(s,0)}
\qquad
P_s=S_{(s,0)}Z_{(s,1)}.
\end{equation}
After the transformation, these operators are replaced by
\begin{equation}
\begin{aligned}
\tilde{Q}_a=UQ_aU^{\dagger}&=X_{\widetilde{(a+1,0)}}^{(a,1)}Q_a
&\qquad\quad
\tilde{P}_a=UP_aU^{\dagger}&=Z_{\widetilde{(a,0)},\widetilde{(a+1,0)}}P_a\\
\tilde{Q}_{a+1}=UQ_{a+1}U^{\dagger}&=X_{\widetilde{(a,0)}}^{(a,1)}X_{\widetilde{(a+2,0)}}^{(a+2,1)}Q_{a+1}
&\qquad\quad
\tilde{P}_{a+1}=UP_{a+1}U^{\dagger}&=P_{a+1}\\
\tilde{Q}_{a+2}=UQ_{a+2}U^{\dagger}&=X_{\widetilde{(a+1,0)}}^{(a+2,1)}Q_{a+2}
&\qquad\quad
\tilde{P}_{a+2}=UP_{a+2}U^{\dagger}&=Z_{\widetilde{(a+1,0)},\widetilde{(a+2,0)}}P_{a+2},
\end{aligned}
\label{eq:QP_transformed}
\end{equation}
for a generic $a\in\mathcal{A}$.

\begin{figure}
\subfloat[][Transformed link operators from Eq.~\eqref{eq:QP_transformed}. The gray controls are controls or anti-controls depending on the site index being even or odd.]{\includegraphics{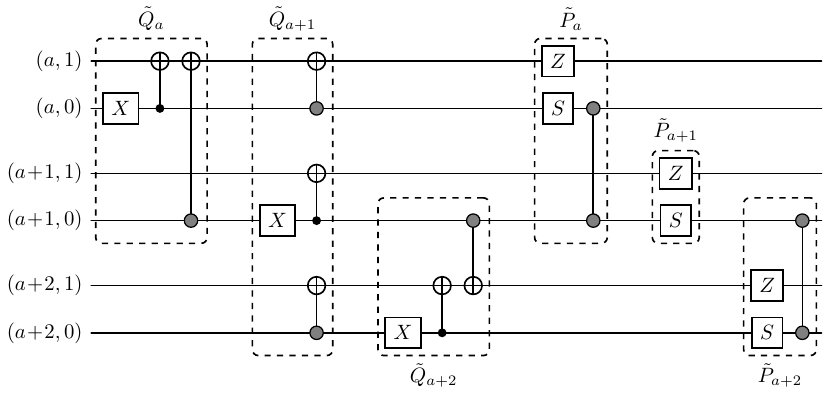}}\\
\subfloat[][Tilde notation: explicit circuits for the gray-controls notation.]{\includegraphics{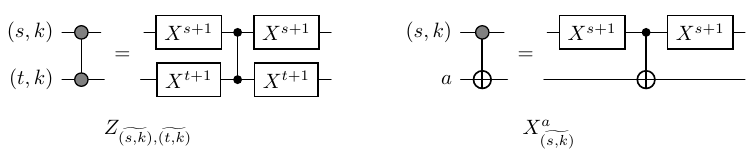}}
\caption{\label{fig:QEC:transformed-link-operators}
Circuital representations of the transformed link operators.
The big gray controls reflect the tilde notation introduced in Eq.~\eqref{eq:1d:tilde-notation}, and they denote anticontrols when on even links, controls when on odd links.
}
\end{figure}

The circuital representations of the transformed link operators are shown in Figure~\ref{fig:QEC:transformed-link-operators}, where we use the big gray controls to reflect the tilde notation introduced in Eq.~\eqref{eq:1d:tilde-notation}. It is easy to see that a controlled-$Z$ operator is added to every operator $P_s$, while a CNOT is added to every $Q_s$, apart when $s\mod 3=1 $, in which case $Q_s$ gets $2$ CNOTs while $P$ is left unchanged. However, these operators are not the gauge invariant operators that appear in the Hamiltonian. In particular, the $Q_s$ operators always appear in the hopping terms as $\psi_{a+1}^{\dagger}\tilde{Q}_a\psi_a$, and when considering the full gauge invariant operator, some of the $\tilde{Q}_s$ operators can be simplified.

Looking at its circuit representation, we see that $\tilde{Q}_a$ contains two CNOTs targeting the qubit $(a,1)$.
As we have seen, this qubit is in a repetition code with the other two qubits $(a\p1,1)$ and $(a\p2,1)$, so it cannot be flipped by a gauge invariant operator without flipping the three qubits together.
As $\tilde{Q}_a$ does not act on $(a\p1,1)$ and $(a\p2,1)$, it cannot flip $(a,1)$ either, and the only possibility is that the product of CNOTs in $\tilde{Q}_a$ must act trivially.
As a sanity check, it can be seen that whenever the two CNOTs act nontrivially, the state is annihilated by $\psi_{a+1}^{\dagger}$ or $\psi_a$.
We conclude that, when applying the hopping term $\psi_{a+1}^{\dagger}\tilde{Q}_a\psi_a$, the two CNOTs in $\tilde{Q}_a$ can be removed.
The same considerations hold for $\tilde{Q}_{a+2}$, but not for $\tilde{Q}_{a+1}$. This means that, when considering the gauge invariant hopping term, the three CNOTs in $\tilde{Q}_{a+1}$ either activate simultaneously, or none of them activate.

It is interesting to notice that the total number of gates required for the three hopping terms on links $a$, $a\p1$ and $a\p2$ is conserved by the change of basis $U$ in the case of $\mathbb{Z}_4$. Indeed, we started from the three $Q_s$, each one having a CNOT, and we ended up with the three $\tilde{Q}_s$, two of which have no CNOTs, and one of them has three.
In order to use the gauge constraints to correct bit-flip errors, the only price we need to pay is the two extra controlled-$Z$ every three $P$ operators shown in Figure~\ref{fig:QEC:transformed-link-operators}.

To conclude, we comment on what happens when the number of sites is not a multiple of three.
The case $V\mod3=0$ is ideal because every site $s\notin\mathcal{A}$ is neighbour to a site $a\in\mathcal{A}$, so the operators $U_{s-1\leftarrow s}$ in Eq.~\eqref{eq:QEC:one-step-left} and $U_{s\rightarrow s+1}$ in Eq.~\eqref{eq:QEC:one-step-right} work on every point of the lattice.
If the number of sites is not a multiple of three, we can still choose $\mathcal{A}$ to be the subset of sites labeled by the multiples of 3, $\mathcal{A}=\{s\,:\,s\mod3=0\}$.
In the case $V\mod3=1$ however, the site $V\m2$ is two sites away from the nearest sites in $\mathcal{A}$, $a=0$ and $a=V\m4$.
We then need to introduce the operators
\begin{equation}
U_{s\rightarrow s+2}=
X_{\widetilde{(s-1,0)}\widetilde{(s,0)}}^{(s+1,1)}
X_{\widetilde{(s,0)}\widetilde{(s+1,0)}}^{(s+1,1)}
U_{s\rightarrow s+1},
\end{equation}
which can be checked to move all the controlled-$Z$ operators in $g_{s,1}$ and $g_{s+1,1}$ to $g_{s+2,1}$.
Then, in Eq.~\eqref{eq:QEC:change-of-basis}, we replace $U_{V-1\rightarrow0}$ with $U_{V-2\rightarrow0}$.
Doing so, and considering again only the Pauli generators after the transformation, it turns out that the qubits $(V\m4,1)$ to $(V\m1,1)$ end up in a four-qubit repetition code.

Finally, in the case $V\mod3=2$, the sites $V\m3$ and $V\m2$ are two sites away from $\mathcal{A}$.
We introduce also
\begin{equation}
U_{s-2\leftarrow s}=
X_{\widetilde{(s-2,0)}\widetilde{(s-1,0)}}^{(s-2,1)}
X_{\widetilde{(s-1,0)}\widetilde{(s,0)}}^{(s-2,1)}
U_{s-1\leftarrow s},
\end{equation}
and, in Eq.~\eqref{eq:QEC:change-of-basis}, we replace $U_{V-3\leftarrow V-2}$ with $U_{V-5\leftarrow V-3}$ and $U_{V-1\rightarrow0}$ with $U_{V-2\rightarrow0}$.
Then, the qubits $(V\m5,1)$ to $(V\m1,1)$ end up in a five-qubit repetition code.
However, the operators on the links $V\m5$ to $V\m1$ get an increased nonlocality in general.

\subsection{Error correction for generic $\mathbb{Z}_{2^\eta}$ in one dimension}
\label{subsec:QEC-1d-generic-eta}

In the previous section we explicitly showed how one can build an error correcting code by exploiting the gauge symmetry, for a $\mathbb{Z}_4$ LGT. Now we will generalize such a result for higher values of the gauge field truncations, considering $\mathbb{Z}_{2^\eta}$, for $\eta > 2$.

Let us define a transformation of the form of $U$ in Eq.~\eqref{eq:QEC:change-of-basis} for every level,
\begin{equation}
\begin{gathered}
U^{(k)}=\prod_{a\in\mathcal{A}}U_{a-1\rightarrow a}^{(k)}U_{a\leftarrow a+1}^{(k)},\\
U_{s-1\leftarrow s}^{(k)}=X_{\widetilde{(s-1,k-1)}\widetilde{(s,k-1)}}^{(s-1,k)}
\qquad
U_{s\rightarrow s+1}^{(k)}=X_{\widetilde{(s-1,k-1)}\widetilde{(s,k-1)}}^{(s,k)}.
\end{gathered}
\end{equation}
With this notation, the transformation in Eq.~\eqref{eq:QEC:change-of-basis} would be $U=U^{(1)}$.
The Toffoli operators of $U^{(k)}$ have targets on the $k$\textsuperscript{th} level, and controls on the $(k\m1)$\textsuperscript{th} level.
For this reason, $U^{(k)}$ does not commute with $U^{(k+1)}$.
To extend the result we obtained for $\eta=2$ to larger $\eta$, we apply first the highest-level transformation $U^{(\eta-1)}$.
Repeating the same steps that led us to Eq.~\eqref{eq:QEC:transformed-generators}, we get the same result on the $(\eta\m1)$\textsuperscript{th} level.
Given $a\in\mathcal{A}$, the generators $g_{a+1,\eta-1}$ and $g_{a+2,\eta-1}$ have been turned into the generators of a repetition code, while the transformed $g_{a,\eta-1}$ carries the controlled-$Z$ operators.
All the generators on the lower levels commute with $U^{(\eta-1)}$, and are left unmodified.

At this point, we can discard the non-Pauli generators on the highest level, the transformed $g_{a,\eta-1}$ with $a\in\mathcal{A}$, as we do not need them for error correction, and we do not care about how they are affected by the next transformations.
Moreover, the transforms of $g_{a+1,\eta-1}$ and $g_{a+2,\eta-1}$ act only on the level $(\eta\m1)$, and we can proceed with $U^{(\eta-2)}$ without affecting them.
The result is again the same as for the highest level. The generators $g_{a+1,\eta-2}$ and $g_{a+2,\eta-2}$ are turned into the generators of a repetition code, while the transformed $g_{a,\eta-2}$ carries the controled-$Z$ operators.
Again, we can discard the non-Pauli generators on the level $\eta\m2$.

We see that we can iterate the procedure in this way until the first level.
The full transformation is
\begin{equation}
U_{(\eta)}=\prod_{k=1}^{\eta-1}U^{(k)}=U^{(1)}\cdots U^{(\eta-1)},
\end{equation}
where the order of multiplication matters.
It transforms the generators into $\tilde{g}_{s,k}=U_{(\eta)}g_{s,k}U_{(\eta)}^{\dagger}$, where the Pauli generators are
\begin{equation}
\begin{aligned}
\tilde{g}_{s,0}&=g_{s,0}&&\forall s,\\[4mm]
\tilde{g}_{s,k}&=Z_{(s-1,1)}Z_{(s,1)}&&
s\notin\mathcal{A},\quad k=1,\dots,\eta-1,
\end{aligned}
\label{eq:QEC:transformed-Pauli-larger-eta}
\end{equation}
and the non-Pauli generators are $\tilde{g}_{a,k}$ with $a\in\mathcal{A}$.
As in the case of $\eta=2$, the Pauli generators in Eq.~\eqref{eq:QEC:transformed-Pauli-larger-eta} form a distance-three classical code, made of one copy of the $\mathbb{Z}_2$ gauge-covariant code of Reference~\cite{Spagnoli-Roggero-Wiebe_2024-05}, together with many copies of the repetition code.

The parameters of the code can be easily computed counting the number of Pauli stabilizers. In a $1$-dimensional lattice with $V$ sites, we have $V+\eta V$ qubits to encode the physical system (considering $\eta$ qubits per link to encode the $2^\eta$ levels of the gauge field). Then, for the level $k=0$, we have one constraint per site, while for the higher levels, we have $2$ constraints every $3$ links, which means $[(\eta+1) V, (\eta+2)V/3 , 3]$.

Notice that, in order to construct this error correcting code, we did not use the entire gauge symmetry.
Indeed, the transformed stabilizers $\tilde{g}_{s,k}$ that are non-Pauli have not been used to construct the error-correcting code.
This means that the space stabilized by this error-correcting code is larger than the physical space stabilized by the Gauss operators $G_s$.
This implies that the logical operations of the error-correcting code are not all gauge-invariant in general.
The gauge invariant operators (which appear in the Hamiltonian) must be the particular linear combinations of logical operations that commute with all the stabilizers of the code, which is given by construction, and also with all the non-Pauli $\tilde{g}_{s,k}$ we did not use for error correction.
This is one of the main differences from the case of $\mathbb{Z}_2$, where the logical degrees of freedom were gauge invariant by construction.
Moreover, this explains how it is possible that, despite the physical space having the dimensionality which is not a power of two, we found a Pauli error-correcting code to stabilize it.
The reason is that we enlarged the stabilized space to make it become a power of two.
The Pauli generators $\tilde{g}_{s,k}$ alone cannot correct all the gauge-violating errors, but only single-qubit bit flips.
However, there is still the possibility to use the remaining non-Pauli ones to detect all the gauge violations, even though we do not explore this here.
In any case, in absence of errors, the evolution under the gauge-invariant Hamiltonian preserve all the constraints, both Pauli and non-Pauli.

\subsection{Error correction for $\mathbb{Z}_4$ in two dimensions}
\label{subsec:QEC-2d}

The techniques discussed in this section can be extended to higher dimensions.
In one dimension, we have seen that the strategy relies on the possibility of moving controlled-$Z$ operators from a generator $g_{s,k}$ to a neighbour, say $g_{s+1,k}$, level by level.
The generators that receive the controlled-$Z$ operators are then ignored for the purposes of error correction, while the Pauli generators form a distance-$3$ error-correcting code.
In order to achieve this result, we used two facts:
\begin{itemize}
    \item the Pauli operator $Z_{(s,k)}$ appears in both $g_{s,k}$ and $g_{s+1,k}$;
    \item all controlled-$Z$ operators appearing in a level-$k$ generator act only on levels lower than $k$.
\end{itemize}

All these features are still true in higher dimensions.
For simplicity, we focus now on the case of $\mathbb{Z}_4$ in two dimensions, but everything can be generalized to $\eta>2$ in the way discussed in Subsection~\ref{subsec:QEC-1d-generic-eta}.
Furthermore, the techniques used in Subsection~\ref{subsec:binary-stabilizer-2d} to find the stabilizer group in more than one dimension suggest that the same idea can be used for larger $\eta$ or in higher dimensions.
The following discussion involves only the first-level generators.
The zeroth-level generators are still there, but they are never affected by our manipulations.
Consider two neighbouring generators $g_{\mathbf{s},1}$ and $g_{\mathbf{t},1}$.
Denoting by L the link connecting the sites $\mathbf{s}$ and $\mathbf{t}$, and looking at Eq.~\eqref{eq:2d:generator-structure}, we can schematically write the two generators as
\begin{equation}
g_{\mathbf{s},1}=
Z_{(\text{L},1)}
\otimes\zeta_{\mathbf{s}}
\otimes C_{\mathbf{s},1},
\qquad
g_{\mathbf{t},1}=
Z_{(\text{L},1)}
\otimes\zeta_{\mathbf{t}}
\otimes C_{\mathbf{t},1},
\end{equation}
where $\zeta_{\mathbf{s}}$ and $\zeta_{\mathbf{t}}$ represent the $Z$ operators on the links different from L.
In the following, we need to attach targets to controlled-$Z$ operators.
By this, in general we mean that the Toffoli gate obtained by attaching a target on a qubit $c$ to a controlled-$Z$ operator $Z_{a,b}$ is the operator $X_{a,b}^c$.
Consider now the operator $U_{\mathbf{s}\rightarrow\mathbf{t}}$, a product of Toffoli operators obtained by attaching a target on the qubit $(L,1)$ to each controlled-$Z$ in $C_{\mathbf{s},1}$.
Using again the identity $X_{a,b}^cZ_cX_{a,b}^c=Z_cZ_{a,b}$, we can see that $U_{\mathbf{s}\rightarrow\mathbf{t}}$ moves the controlled-$Z$ operators from $g_{\mathbf{s},1}$ to $g_{\mathbf{t},1}$,
\begin{equation}
\begin{aligned}
U_{\mathbf{s}\rightarrow\mathbf{t}}
\,g_{\mathbf{s},1}\,
U_{\mathbf{s}\rightarrow\mathbf{t}}^{\dagger}&=
Z_{(\text{L},1)}
\otimes\zeta_{\mathbf{s}},\\
U_{\mathbf{s}\rightarrow\mathbf{t}}
\,g_{\mathbf{t},1}\,
U_{\mathbf{s}\rightarrow\mathbf{t}}^{\dagger}&=
Z_{(\text{L},1)}
\otimes\zeta_{\mathbf{s}}
\otimes C_{\mathbf{s},1}C_{\mathbf{t},1}.
\end{aligned}
\end{equation}
We can do the same with the other three neighbours of $\mathbf{t}$, exploiting the $Z$ operators in $\zeta_{\mathbf{t}}$.
Once this is done, the transformed $g_{\mathbf{t},1}$ carries all the controlled-$Z$ operators from itself and its neighbours, and all its neighbours are Pauli operators.
We call $g_{\mathbf{t},1}$ an accumulation point.

\begin{figure}
\includegraphics{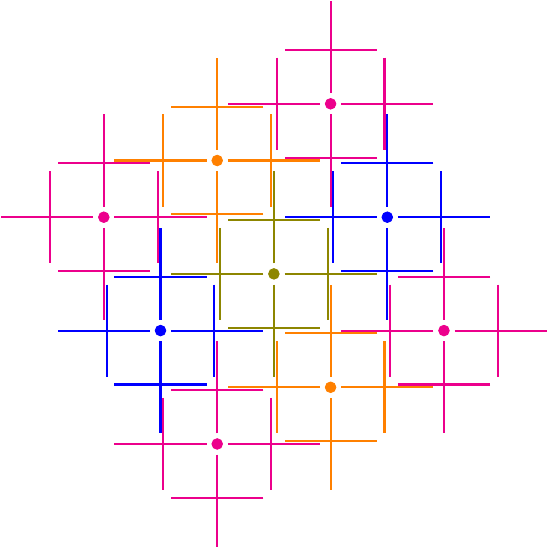}
\caption{
Accumulation points and first-level Pauli generators in the two-dimensional $\mathbb{Z}_4$ theory.
The dots are accumulation points, and the crosses are Pauli generators formed by four $Z$ operators, one per leg.
}
\label{fig:QEC:accumulation-code-2d}
\end{figure}

The task now is to choose an appropriate set of accumulation points.
In one dimension, choosing one site every three allows us to obtain a distance-three classical code.
In two dimensions, we can still obtain a distance-three classical code by choosing the set $\mathcal{A}$ in a careful way, as shown in Figure~\ref{fig:QEC:accumulation-code-2d}.
A dot represents an accumulation point, and the crosses of the same colour around it are its neighbours.
The sites are located at the dots and at the centers of the crosses.
After all the controlled-$Z$ have been accumulated on the accumulation points, each cross is the product of four $Z$ operators on the four first-level qubits of the links covered by the cross.

For the purposes of error correction, we need to only focus on the crosses and disregard the generators on the dots.
The crosses form a distance-three classical code, with a simple error syndrome.
First, notice that each cross has exactly one leg not overlapping with another cross.
The other three legs overlap with three different crosses.
Suppose now that there is a single bit flip on a first-level qubit, and we measure the generators represented by the crosses.
The only possible outcomes are that either a single cross gives $-1$, or two neighbouring crosses with an overlapping leg give $-1$.
In the first case, the error is located on the leg not overlapping with other crosses.
In the second case, the error is located on the leg shared by the two crosses.
We conclude that each single-qubit flip on the first level is uniquely identified by the syndrome and can be corrected.

This shows that, after the transformations $U_{\mathbf{s}\rightarrow\mathbf{t}}$, the first-level $g$ operators that are not on the accumulation points form a $[2V, \frac{6}{5}V, 3]$ code on the first level (assuming $L$ to be a multiple of $5$).
On the zeroth level, we still have the $g_{\mathbf{s},0}$ generators.
They form the $\mathbb{Z}_2$ gauge-covariant code of Reference~\cite{Spagnoli-Roggero-Wiebe_2024-05}, and they can be used to correct bit flips on the zeroth level.  Therefore, we can correct any single-qubit bit flip on the lattice by using only the gauge constraints.
If $L$ is not a multiple of 5, the links $(\mathbf{s},l)$ with $s_1$ or $s_2$ between about $L-5$ and $L-1$ should be taken care of in a similar way as at the end of Subsection~\ref{subsec:QEC-1d-eta2}.

\section{Application: gauge fixing}
\label{sec:gauge-fixing}

In this section, by proceeding in a similar way as in Section~\ref{sec:QEC}, we find a unitary transformation $U$ that removes the gauge redundancies and reduces the number of qubits for a simulation.
The idea is that we can turn an operator in the stabilizer group into a $Z$ operator acting on a single qubit.
The qubit is then frozen in the $\ket{0}$ state and can be discarded.
Again, we show how to do that in detail for $\mathbb{Z}_4$ in one dimension, but the method can be generalized to larger gauge groups and to higher dimensions.

For convenience, we rewrite the stabilizers of $\mathbb{Z}_4$ once again,
\begin{equation}
g_{s,0}=(-1)^sZ_{(s-1,0)}Z_sZ_{(s,0)},\qquad
g_{s,1}=Z_{(s-1,1)}Z_{(s,1)}Z_{\widetilde{(s-1,0)}\widetilde{(s,0)}},
\end{equation}
However, this time we choose a different, although equivalent, set of generators.
We keep all the generators $g_{s,0}$ and the generators $g_{s,1}$ with $s=1,\dots,V\m1$ while we replace
\begin{equation}
g_{0,1}\longrightarrow
g=
\prod_{s=0}^{V-1}g_{s,1}=
\prod_{s=0}^{V-1}Z_{\widetilde{(s-1,0)}\widetilde{(s,0)}}.
\end{equation}

First, we can remove the qubits associated with the sites by applying the transformation
\begin{equation}
U_0=\prod_sU_{s,0},\qquad
U_{s,0}=X_{(s-1,0)}^sX_{(s,0)}^s(X_s)^s,
\label{eq:gauge-fixing:U0-transformation}
\end{equation}
acting only on the zeroth level.
A $U_{s,0}$ operator commutes with both $U_{t,0}$ and $g_{t,0}$ if $t\ne s$, with $g$, and with $g_{t,1}$ for any $t$.
Then, using the identities $X_a^bZ_bX_a^b=Z_aZ_b$ and $Z_a^2=\mathbb{1}$, we have
\begin{equation}
U_0g_{s,0}U_0^{\dagger}=U_{s,0}g_{s,0}U_{s,0}^{\dagger}=Z_s.
\label{eq:gauge-fixing:transformed-generators-under-U0}
\end{equation}
As the generators $g_{s,0}$ are turned into a $Z$ operator acting on the single qubit of site $s$, these qubit are frozen in the $\ket{0}$ state after the transformation and can be discarded.

We are now left with only the link qubits and the constraints $g_{s,1}$ and $g$.
Then, we apply the transformation
\begin{equation}
U_1'=
\prod_{s=1}^{V-1}X_{(s-1,1)}^{(s,1)}=
X_{(0,1)}^{(1,1)}X_{(1,1)}^{(2,1)}\cdots X_{(V-2,1)}^{(V-1,1)},
\end{equation}
where again the order of multiplication matters.
This transformation acts only on the first-level qubits, and thus it commutes with $g$.
On the other generators, its effect is to remove one of the two $Z$ operators acting on the first level,
\begin{equation}
U_1'g_{s,1}{U_1'}^{\dagger}=
Z_{(s,1)}Z_{\widetilde{(s-1,0)}\widetilde{(s,0)}}=
g^{(1)}_{s,1}.
\end{equation}

For every $g^{(1)}_{s,1}$, consider the Toffoli operator $X_{\widetilde{(s-1,0)}\widetilde{(s,0)}}^{(s,1)}$ obtained by attaching a target on qubit $(s,1)$ to the controlled-$Z$ operator $Z_{\widetilde{(s-1,0)}\widetilde{(s,0)}}$.
We use these Toffoli operators to remove the controlled-$Z$ operators applying
\begin{equation}
U_1''=
\prod_{s=1}^{V-1}X_{\widetilde{(s-1,0)}\widetilde{(s,0)}}^{(s,1)},
\end{equation}
whose action is
\begin{equation}
g^{(2)}_{s,1}=U_1''g_{s,1}'{U_1''}^{\dagger}=
Z_{(s,1)}.
\label{eq:gauge-fixing:Toffoli-product}
\end{equation}
Therefore, all the qubits $(s,1)$ with $s=1,\dots,V-1$ are frozen in the $\ket{0}$ state and can be discarded.
The Toffoli operators in Eq.~\eqref{eq:gauge-fixing:Toffoli-product} are similar to the ones in Eq.~\eqref{eq:QEC:change-of-basis}.
The difference is that now each generator $g_{s,1}'$ carries only one $Z$ operator on the qubit $(s,1)$, as a consequence of applying the transformation $U_1'$.
The $Z_{(s,1)}$ is not shared among different generators, in the sense that it appears only in $g_{s,1}^{(1)}$.
This is why, here, the controlled-$Z$ operators are removed, and not moved, by the Toffoli operators.
Before applying $U_1'$, each $Z_{(s,1)}$ operator on the first level is shared between the two generators $g_{s,1}$ and $g_{s+1,1}$, which is why in Subsection~\ref{sec:QEC} the controlled-$Z$ operators are moved from one generator to another.

We summarize the final situation after applying the full transformation $U=U_1''U_1'U_0$.
Apart from the lower link qubits $(s,0)$, only the qubit $(0,1)$ is left, and it is now unconstrained.
The only remaining constraint is $g$, which commutes with all the transformations.
More than half of the eigenvalues of this operator are 1, which implies that more than half of the states in the new, total Hilbert space are physical.
Therefore, we cannot remove any more qubits.
If we did so, the total Hilbert space would no longer be large enough to host the physical space.

We turn now our attention to the gauge-invariant operators.
Their form after the transformation is obtained by conjugation, as we did for the stabilizer generators.
Furthermore, all the operators that, after the transformation, act on the discarded qubits can be removed as they must act trivially.
Since the calculations are somewhat tedious, we only provide the final results.
The hopping terms preserve their locality.
Setting $B_s=\psi_{s+1}^{\dagger}Q_s\psi_s$, we find
\begin{equation}
\begin{aligned}
UB_0U^{\dagger}&=\mathbb{P}_{V-1,0}^-Q_0\mathbb{P}_{0,1}^-\\
UB_sU^{\dagger}&=\mathbb{P}_{s-1,s}^-X_{(s,0)}\mathbb{P}_{s,s+1}^-&\quad&\text{if $s$ is even}\\
UB_sU^{\dagger}&=\mathbb{P}_{s-1,s}^+X_{(s,0)}\mathbb{P}_{s,s+1}^+&\quad&\text{if $s$ is odd}
\end{aligned}
\qquad\qquad
\mathbb{P}_{s,t}^{\pm}=\frac{1}{2}\left(\mathbb{1}\pm Z_{(s,0)}Z_{(t,0)}\right).
\end{equation}
The electric terms loose their locality, acquiring a chain of controlled-$Z$ operators connecting them to the link 0,
\begin{equation}
\begin{aligned}
UP_0U^{\dagger}=P_0,\qquad
UP_sU^{\dagger}=S_{(s,0)}\prod_{t=1}^sZ_{\widetilde{(t-1,0)}\widetilde{(t,0)}}.
\end{aligned}
\end{equation}
Notice that the transformation $U_0$ alone commutes with the $P$ operators and preserves the locality of the hopping terms, allowing us to remove the site qubits without spoiling locality.

If $\eta>2$, we can proceed in order from the highest level to the lowest level, similarly to what is done in Subsection~\ref{subsec:QEC-1d-generic-eta} to build the error correcting code for generic $\eta$.
Level by level, we can repeat all the steps described in this subsection.
First, we choose a different, although equivalent, set of generators, replacing
\begin{equation}
g_{0,k}\longrightarrow
g_k=
\prod_{s=0}^{V-1}g_{s,k}=
\prod_{s=0}^{V-1}Z_{\widetilde{(s-1,k-1)}\widetilde{(s,k-1)}},\qquad
k=1,\dots,\eta-1.
\end{equation}
Second, we define $U_k=U_k''U_k'$, with
\begin{equation}
\begin{aligned}
U_k'&=
\prod_{s=1}^{V-1}X_{(s-1,k)}^{(s,k)}=
X_{(0,k)}^{(1,k)}X_{(1,k)}^{(2,k)}\cdots X_{(V-2,k)}^{(V-1,k)},\\
U_k''&=
\prod_{s=1}^{V-1}X_{\widetilde{(s-1,k-1)}\widetilde{(s,k-1)}}^{(s,k)}.
\end{aligned}
\end{equation}
Finally, we put all the transformations together in the right order,
\begin{equation}
U_{(\eta)}=\prod_{k=0}^{\eta-1}U_k=U_0U_1\cdots U_{\eta-1},
\end{equation}
where we recall that $U_0$ has a different form from $U_k$ with $k\ge1$.
The final result on the generators is
\begin{equation}
\begin{aligned}
U_{(\eta)}g_{s,0}U_{(\eta)}^{\dagger}&=Z_s&&\forall s\\
U_{(\eta)}g_{s,k}U_{(\eta)}^{\dagger}&=Z_{(s,k)}&&
s=1,\dots,V-1,\quad k=1,\dots,\eta-1,
\end{aligned}
\end{equation}
which allows us to remove all the site qubits and the link qubits on the higher levels, with the exception of link 0.
The $g_k$ generators are modified by $U$, but their form is not relevant for the purposes of the present manuscript.

An advantage of the method we propose here to remove gauge redundancies is that it is generalizable to higher dimensions. First, one has to find the $d$-dimensional transformation $U_k'$, which will transform the generators $g_{\mathbf{s},k}$ into operators with a single $Z$ operator not shared with other generators. Then, the transformation $U_k''$ needs to be found, by attaching targets on the correct qubits to the controlled-$Z$ operators in the generators.

\section{Conclusions}
\label{sec:conclusions}

In this work we have further developed the connection between gauge theories and quantum error correction at a practical level.
Specifically, we have introduced the binary Gauss stabilizers for the gauge-invariant subspace of $\mathbb{Z}_{2^\eta}$ gauge theories, an alternative to the Gauss operators commonly encountered in the literature of lattice gauge theories.
The new stabilizer set is built from products of Pauli $Z$, controlled-$Z$ and multicontrolled-$Z$ operators, all of which have only $\pm1$ eigenvalues.
We have also shown explicitly two important applications of this stabilizer group, error correction and gauge fixing, and we believe many others are to be discovered.

In order to exploit the gauge redundancies for quantum error correction, in Section~\ref{sec:QEC} we have provided a change of basis that turns a subgroup of the stabilizer set into a Pauli stabilizer group with only $Z$ operators.
The obtained Pauli stabilizers forms a classical code with distance three, which can correct any single-qubit bit-flip error using only the qubits required for the lattice system.
Thus, after the change of basis, one can focus only on the Pauli subgroup to perform error correction, and ignore the remaining non-Pauli generators.
This means that the non-Pauli generators are still there, but they do not play any active role, at least in the context of error correction.
The codespace stabilized by the Pauli subgroup strictly contains the gauge-invariant subspace.
This follows immediately from the fact that the gauge-invariant subspace satisfies both the Pauli and the non-Pauli constraints, while the codespace satisfies only the Pauli ones.
Clearly, the Pauli subgroup does not correct every gauge-violating error (only single-qubit errors), but the remaining non-Pauli ones may still be employed to further detect or prevent gauge violations, for instance through energy penalties or dynamical decoupling schemes (see e.g.~\cite{PhysRevX.3.041018,lamm2020suppressingcoherentgaugedrift,PRXQuantum.2.040311}), a possibility we leave for future investigations.
Furthermore, our work could establish a general strategy to include symmetries in the design of problem-specific error-correcting codes.

The method of gauge fixing discussed in Section~\ref{sec:gauge-fixing} is technically very similar to the method of error correction.
A suitable change of basis turns a subset of the binary Gauss generators into single $Z$ operators.
These single-$Z$ generators freeze the qubits they act on, making them disposable.
The gauge-fixing method we have introduced should be compared to existing approaches, such as the axial gauge, to investigate possible advantages~\cite{Farrell-Chernyshev-Powell-Zemlevskiy-Illa-Savage_22-07,Farrell-Illa-Ciavarella-Savage_23-08}.

The gauge constraints of a $\mathbb{Z}_N$ theory, being diagonal in the computational basis, allow us to correct only off-diagonal errors, such as bit-flips.
Thus, as in the $\mathbb{Z}_2$ case, handling phase-flip errors inevitably requires additional resources.
A possible way to obtain a full quantum error-correcting code is by concatenation with, for instance, a phase-flip repetition code, as done in previous works on this topic~\cite{rajput2023quantum,Spagnoli-Roggero-Wiebe_2024-05,Spagnoli:2026qni}.
Another viable approach is to identify the logical degrees of freedom in the codespace stabilized by the Pauli subgroup, and map these degrees of freedom to a quantum code.
For example, consider three qubits in a repetition code in the lower part of Figure~\ref{fig:QEC:accumulating-CZ-1d}, namely after the change of basis that provides the Pauli subgroup.
We could replace these three qubits with, say, seven qubits in the Steane code, or whichever code we deem more suitable.
Thus, an important connection is that we can view our strategy for error correction as a partial gauge fixing. Once we turn some generators into Pauli operators, it is easy to find the logical degrees of freedom, and considering the logical degrees of freedom is equivalent to removing the constraints associated with the Pauli operators.
It is a partial gauge fixing in that it still leaves a number of constraints growing with the system size.

Another direction worth exploring is whether our framework can be leveraged to design more efficient classical-shadow tomography tailored to gauge theories. Classical-shadow protocols are based on random unitary transformations followed by measurements, typically drawn from the Pauli group.
However, for a gauge theory, the number of relevant (gauge invariant) Pauli operators is vanishingly small compared to the total number of Pauli operators on the lattice. Recent work in Ref.~\cite{Bringewatt2026classicalshadows} has tackled this problem in a $\mathbb{Z}_2$ pure gauge theory by exploiting the known duality to the Ising model to directly identify the subset of gauge-invariant operators to sample from. In the language of quantum error correction used in our work, this corresponds to performing shadow tomography on the basis of logical operators. When the stabilized subspace matches exactly the gauge-invariant space, as e.g. in $\mathbb{Z}_2$ theories in any dimension with and without matter~\cite{Spagnoli-Roggero-Wiebe_2024-05}, this allows for the overhead discussed above to be completely removed. On the other hand, both for $\mathbb{Z}_N$ theories for $N$ prime coupled to fermionic matter~\cite{Spagnoli:2026qni}, or in the present case of $\mathbb{Z}_{2^\eta}$ gauge theories with $\eta>1$, the Pauli-stabilized subspace is larger than the gauge-invariant space.
However, by considering only the logical Pauli operators of this subspace, we may still significantly reduce the set of Pauli operators to sample from when implementing classical shadows, reducing the cost of the algorithm.

On a more speculative level, a more fundamental property of these lattice gauge theories seem to emerge: if we consider the standard Gauss operators from Eq.~\eqref{eq:Gauss_operators} we notice that they are manifestly Clifford operations only for $N=2$ and $N=4$, in any spatial dimension.
The binary Gauss stabilizers introduced in Subsec.~\ref{subsec:binary-stabilizer-1d} for the one-dimensional case are all Clifford operations.
This shows that the physical subspace can be described exactly using only Clifford stabilizers in one dimension for any $N=2^{\eta}$.
This does not appear to be the case in higher dimensions: for $\eta>2$ the stabilizer set we construct from our procedure (see also App.~\ref{app:binary-stabilizer-2d-generic-eta} for an explicit calculation in $d=2$ dimensions) always contain non-Clifford operations like the Toffoli gate. It still remains to be formally shown that non-Clifford operations are actually required for $d\geq2$ and $\eta>2$, we conjecture that it is. This is motivated by the empirical evidence that $\mathbb{Z}_N$ lattice gauge theories show a markedly different phase diagram depending on the value of $d$ and $N$: for $d=1$ and any $N$ the theory has only two different phases, a confined and a deconfined phase, while for $d>1$ and a critical value of $N=N_c<8$  a $U(1)$ phase is found between the two~\cite{Elitzur:1979uv,Horn:1979fy,Ukawa:1979yv,Kogut:1980fn,Zarei:2020lrm,Nyhegn:2020cxu,Giansiracusa:2025hfj}.
A more deep exploration of the relationship between the properties of lattice gauge theories and the structure of their corresponding stabilizer set will be the subject of a future publication.

\section{Acknowledgments}
We want to thank P. Hauke, E. Rico, X. Yao and Y. Omar for discussions on the topics covered in this work.

M. T. thanks the support from
CeFEMA -- Centre of Physics and Engineering of Advanced Materials,
through contracts
UID/04540/2025 [https://doi.org/10.54499/UID/PRR/04540/2025],
UID/PRR/04540/2025 [https://doi.org/10.54499/UID/PRR/04540/2025 ], and
\newline
UID/PRR2/04540/2025 [https://doi.org/10.54499/UID/PRR2/04540/2025]
with the Portuguese Agência para a Investigação e Inovação AI2,
and from LaPMET -- Laboratory of Physics for Materials and Emerging Technologies,
through contract
LA/P/0095/2020 [https://doi.org/10.54499/LA/P/0095/2020]
with the Portuguese Agência para a Investigação e Inovação AI2,
as well as from projects QuantHEP and HQCC supported by
the EU QuantERA ERA-NET Cofund in Quantum Technologies and by
FCT -- Funda\c{c}\~{a}o para a Ci\^{e}ncia e a Tecnologia (Portugal) (QuantERA/004/2021).
M. T. also thanks FCT for support through Grant No. PRT/BD/154668/2022.

\appendix

\section{Number of physical states}
\label{app:number_physical_states}

The number of physical states is given by Eq.~\eqref{number-of-physical-states}. In our case, we will consider $V$ to be even, since we have staggered fermions, and $N=2^\eta$. Then, let us rewrite the expression as:
\begin{equation}
D=N^{E-V}\sum_{c=0}^{N-1}\left(2+2\cos\left(\frac{2\pi c}{N}\right)\right)^{V/2}=N^{E-V}2^V\sum_{c=0}^{N-1}\cos\left(\frac{\pi c}{N}\right)^V
\end{equation}
We can now expand the sum using the identity
\begin{equation}
\cos\left(\frac{\pi c}{N}\right)^V=\frac{1}{2^{V}}\sum_{m=0}^{V}\binom{V}{m}\cos\left((V-2m)\frac{\pi c}{N}\right)=\frac{1}{2^{V}}\sum_{k=-V/2}^{V/2}\binom{V}{V/2-k}\cos\left(2k\frac{\pi c}{N}\right)
\end{equation}
We can now perform the sum over $c$ at $m$ fixed using
\begin{equation}
\sum_{c=0}^{N-1}\cos\left(2k\frac{\pi c}{N}\right)=\bigg\{\begin{matrix}N&\text{if}\quad k\!\!\!\mod N=0\\0&\text{otherwise}\end{matrix}\;.
\end{equation}
The sum of the cosine terms then selects $k$ indices that are integer multiples of $N$. But then we have
\begin{equation}
D=N^{E-V+1}\sum_{r\in \mathbb{Z}} \binom{V}{V/2 + rN}\;,
\end{equation}
where the binomial coefficients are defined to be $\binom{V}{j}=0$ for $j\notin [0,V]$.

The first factor $N^{E-V+1}$, for $N=2^\eta$, will always be a power of $2$, and due to the fact that the remaining sum has to result in a positive integer, in order to prove that $D$ is a power of $2$, it is enough to prove that the sum is a power of $2$. For $V=2$ and $N\ge 2$, the sum has only one non-zero term, which is the one with $r=0$, for which the binomial coefficient is just $2$. In this case, $D$ is a power of $2$.

For any $V>2$ and $N>V/2$, again the sum has only one non-zero term, for $r=0$. However, for $V>2$ there exists a prime factor $p$ such that $V/2 \le p \le V$, which divides $\binom{V}{V/2}$. This means that $D$ is not a power of $2$ for $V>2$ if $N>V/2$.

In the other cases, where $V>2$ and $N\le V/2$, it is not straightforward to prove weather $D$ is a power of $2$ from this formula.
However, we run a numerical check, up to $V=600$, for every $N\le V/2$, and we found that no combination $V,N$ exists such that $D$ is a power of $2$.
Thus, we conjecture that $V=2$ is the only case where $D$ is a power of $2$ with $N>2$.

\section{Lemma 1}
\label{app:lemma1}

\begin{lemma}
Let $a_k,b_k,\alpha_k\in\{0,1\}$ for $k=0,\dots,\eta-1$. Then
\begin{equation}
\begin{cases}
a_k\oplus b_k\oplus \alpha_k=0\\[1mm]
\alpha_{k+1}=a_k\alpha_k
\end{cases}
\forall k=0,\dots,\eta-1
\quad\iff\quad
\begin{cases}
a_k\oplus b_k\oplus \alpha_k=0\\[1mm]
\displaystyle
\alpha_k=\alpha_0\oplus\left(\bigoplus_{l=0}^{k-1}\bar{a}_lb_l\right)
\end{cases}
\forall k=0,\dots,\eta-1.
\end{equation}
\label{lemma:ck-equivalence}
\end{lemma}
\begin{proof}
Starting from $\alpha_{k+1}=a_k\alpha_k$, and using $\alpha_k=a_k\oplus b_k$,
we have
\begin{equation}
\alpha_{k+1}\oplus \alpha_k=a_k\alpha_k\oplus \alpha_k=(a_k\oplus1)\alpha_k=
\bar{a}_k(a_k\oplus b_k)=\bar{a}_kb_k.
\end{equation}
Then, since $\alpha_l\oplus \alpha_l=0$, we have
\begin{equation}
\alpha_k=\left(\bigoplus_{l=0}^{k-1}\left(\alpha_l\oplus \alpha_l\right)\right)\oplus \alpha_k=
\alpha_0\oplus\left(\bigoplus_{l=0}^{k-1}\left(\alpha_l\oplus \alpha_{l+1}\right)\right)=
\alpha_0\oplus\left(\bigoplus_{l=0}^{k-1}\bar{a}_lb_l\right).
\end{equation}

Starting from $\alpha_k=\alpha_0\oplus\left(\bigoplus_{l=0}^{k-1}\bar{a}_lb_l\right)$, we find again $\alpha_{k+1}\oplus \alpha_k=\bar{a}_kb_k$.
Then, using $b_k=a_k\oplus\alpha_k$, we obtain $\alpha_{k+1}=a_k\alpha_k$.
\end{proof}

\section{Generators for $\mathbb{Z}_{{2^\eta}}$ in two dimensions}
\label{app:binary-stabilizer-2d-generic-eta}
To find the generator $g_{\mathbf{s},k}$ with generic $k>1$, we follow the same procedure used to determine $g_{\mathbf{s},1}$ in Section~\ref{subsec:binary-stabilizer-2d} of the main text.
We start by first finding circuits computing the carry bits $\alpha_k$ and $\beta_k$.
Second, we conjugate a product of $Z$ operators by these circuits and simplify what we can.
The case with $\eta=3$ is worked out explicitly in Figure~\ref{fig:app:Z8-stabilizer-circuit}.

\begin{figure}[b]
\includegraphics{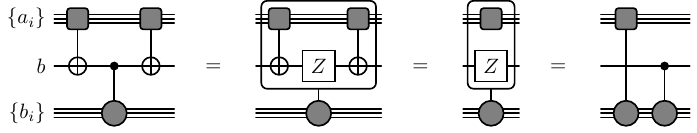}
\caption{
Derivation of the identity $X_{a_1,\dots,a_n}^bZ_{b,b_1,\dots,b_m}X_{a_1,\dots,a_n}^b=Z_{b,b_1,\dots,b_m}Z_{a_1,\dots,a_n,b_1,\dots,b_m}$.
The gray squares and circles represent arbitrary multicontrols.
}
\label{fig:app:general-conjugation-identity}
\end{figure}

We need two registers of ancillary qubits $(\alpha,k)$ and $(\beta,k)$ for the carry bits, with $k=1,\dots,\eta\m1$.
Let us recall the recursion relations in Eq.~\eqref{eq:2d:carries}
\begin{equation}
\alpha_k=\alpha_{k-1}a_{k-1}
\oplus\alpha_{k-1}b_{k-1}
\oplus a_{k-1}b_{k-1}\qquad
\beta_k=\beta_{k-1}c_{k-1}
\oplus\beta_{k-1}d_{k-1}
\oplus c_{k-1}d_{k-1},
\end{equation}
where $\alpha_0=0$, $\beta_0=f_{\mathbf{s}}$ for even $\abs{\mathbf{s}}$ and $\beta_0=1-f_{\mathbf{s}}$ for odd $\abs{\mathbf{s}}$.
Given the action of a Toffoli gate described in Eq.~\eqref{eq:2d:Toffoli-action}, and denoting with $(\alpha, k)$ the qubit storing the $k$-th bit of $\alpha$, $\alpha_k$, we can compute $\alpha_k$ from $\alpha_{k-1}$ by using the unitary $U_{k-1,k}^{(\alpha)}$ such that $U_{k-1,k}^{(\alpha)} \ket{a_{k-1}}\ket{b_{k-1}} \ket{\alpha_{k-1}}\ket{0} = \ket{a_{k-1}}\ket{b_{k-1}}\ket{\alpha_{k-1}}\ket{\alpha_k}$. This unitary can be computed using three Toffoli gates as follows:
\begin{equation}
U_{k-1,k}^{(\alpha)}=
X_{(\alpha,k-1),(A,k-1)}^{(\alpha,k)}
X_{(\alpha,k-1),(B,k-1)}^{(\alpha,k)}
X_{(A,k-1),(B,k-1)}^{(\alpha,k)}.
\end{equation}
Following the notation in the main text, the registers $A$ and $B$ encode the east and north links attached to a site $\mathbf{s}$, $C$ and $D$ encode the west and south links.
Here we focus on a site $\mathbf{s}$ with even $\abs{\mathbf{s}}$.
We will mention at the end the due modifications for odd $\abs{\mathbf{s}}$.
Since $\alpha_0=0$, we do not need to employ an actual qubit to store it, thus
\begin{equation}
U_{0,1}^{(\alpha)}=
X_{(A,0),(B,0)}^{(\alpha,1)}.
\end{equation}
Similarly, for $\beta_k$ we have
\begin{equation}
U_{k-1,k}^{(\beta)}=
X_{(\beta,k-1),(C,k-1)}^{(\beta,k)}
X_{(\beta,k-1),(D,k-1)}^{(\beta,k)}
X_{(C,k-1),(D,k-1)}^{(\beta,k)},
\end{equation}
and, since $\beta_0=f_{\mathbf{s}}$,
\begin{equation}
U_{0,1}^{(\beta)}=
X_{\mathbf{s},(C,0)}^{(\beta,1)}
X_{\mathbf{s},(D,0)}^{(\beta,1)}
X_{(C,0),(D,0)}^{(\beta,1)}.
\end{equation}
The full circuit to compute $\alpha_k$ and $\beta_k$ can be expressed as $U_k=U_k^{(\alpha)}U_k^{(\beta)}$, with
\begin{equation}
\begin{aligned}
U_k^{(\alpha)}&=
\prod_{l=k}^1U_{l-1,l}^{(\alpha)}=
U_{k-1,k}^{(\alpha)}
U_{k-2,k-1}^{(\alpha)}
\cdots
U_{0,1}^{(\alpha)},\\
U_k^{(\beta)}&=
\prod_{l=k}^1U_{l-1,l}^{(\beta)}=
U_{k-1,k}^{(\beta)}
U_{k-2,k-1}^{(\beta)}
\cdots
U_{0,1}^{(\beta)}.
\end{aligned}
\end{equation}
The circuits $U_2^{(\alpha)}$ and $U_2^{(\beta)}$ are shown on the left of the $Z$ operators in the left-hand side of Figure~\ref{fig:app:Z8-stabilizer-circuit}.
Following the discussion in Section~\ref{subsec:binary-stabilizer-2d}, the intermediate form of the generator is
\begin{equation}
\begin{aligned}
g_{\mathbf{s},k}'&=
U_k^{\dagger}
Z_{(A,k)}
Z_{(B,k)}
Z_{(C,k)}
Z_{(D,k)}
Z_{(\alpha,k)}
Z_{(\beta,k)}
U_k\\
&=
Z_{(A,k)}
Z_{(B,k)}
Z_{(C,k)}
Z_{(D,k)}
U_k^{\dagger}
Z_{(\alpha,k)}
Z_{(\beta,k)}
U_k\\
&=
Z_{(A,k)}
Z_{(B,k)}
Z_{(C,k)}
Z_{(D,k)}\left(
{U^{(\alpha)}}_k^{\dagger}
Z_{(\alpha,k)}
U^{(\alpha)}_k\right)\left(
{U^{(\beta)}}_k^{\dagger}
Z_{(\beta,k)}
U^{(\beta)}_k\right)\\
\end{aligned}
\label{eq:app:generic-generator-with-ancillae}
\end{equation}
The intermediate $g_{\mathbf{s},2}'$ is shown on the left-hand side of Figure~\ref{fig:app:Z8-stabilizer-circuit}.

What we need to do now is to simplify this operator by working out the conjugations of $Z$, controlled-$Z$ and multicontrolled-$Z$ by Toffoli gates.
We can focus on the $\alpha$ and $\beta$ parts separately, as is clear from Eq.~\eqref{eq:app:generic-generator-with-ancillae}.
To do this, we need to generalize the identity $X_{a,b}^cZ_cX_{a,b}^c=Z_cZ_{a,b}$ shown in Figure~\ref{fig:2d:conjugation-identity} to deal with controlled-$Z$ and multicontrolled-$Z$ operators.
First, notice that we can add arbitrary controls to the Toffoli gates,
\begin{equation}
X_{a_1,\dots,a_n}^bZ_bX_{a_1,\dots,a_n}^b=Z_bZ_{a_1,\dots,a_n}
\label{eq:app:multi-Toffoli-Z-conjugation}
\end{equation}
Then, as shown through the circuit identities in Figure~\ref{fig:app:general-conjugation-identity}, we can further generalize the formula as
\begin{equation}
X_{a_1,\dots,a_n}^bZ_{b,b_1,\dots,b_m}X_{a_1,\dots,a_n}^b=Z_{b,b_1,\dots,b_m}Z_{a_1,\dots,a_n,b_1,\dots,b_m}.
\end{equation}
The gray squares and circles in the Figure represent arbitrary multicontrols on the set of qubits they are depicted on, and we use two different symbols to denote two, possibly different, sets of controls.
In the second and third pieces of the Figure, the boxes are operators controlled on the qubits below the boxes.
The first and last equalities are common relations between controlled operators.
The second equality follows directly from Eq.~\eqref{eq:app:multi-Toffoli-Z-conjugation}.
An example of the general formula is shown in Figure~\ref{fig:app:general-conjugation-identity-example}.
A mnemonic rule for the formula is that the conjugated operator $Z_{b,b_1,\dots,b_m}$ remains, and gets multiplied by an operator constructed by merging together the indices of $Z_{b,b_1,\dots,b_m}$ and $X_{a_1,\dots,a_n}^b$, where the subscript index $b$ is "contracted" with the superscript index $b$, meaning that the resulting operator will have both $a_i$ and $b_i$ indices but not $b$.

\begin{figure}
\includegraphics{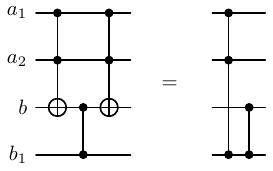}
\caption{Circuital representation of the identity $X_{a_1,a_2}^bZ_{b,b_1}X_{a_1,a_2}^b=Z_{b,b_1}Z_{a_1,a_2,b_1}$.}
\label{fig:app:general-conjugation-identity-example}
\end{figure}

\begin{figure}
\includegraphics{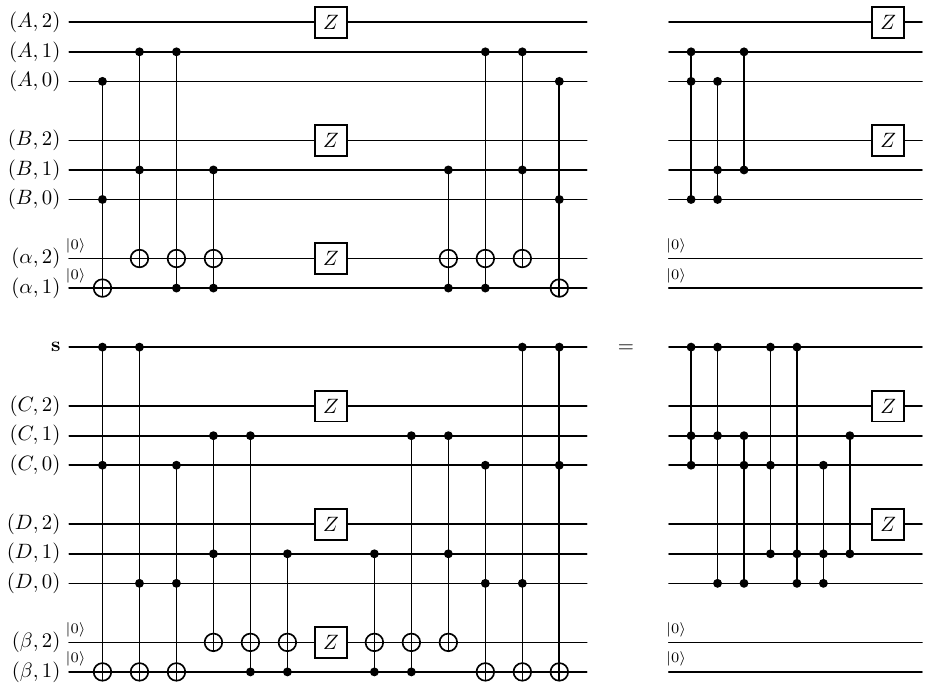}
\caption{
Second-level generator for the stabilizer group of the $\mathbb{Z}_8$ theory in two dimensions.
The ancillary qubits for $\alpha$ and $\beta$ are initially set in the $\ket{0}$ state.
The operator on the left-hand side is $g_{\mathbf{s},2}'$, and the one on the right-hand side is $g_{\mathbf{s},2}$. Of course, the equality holds only if the ancillary qubits start in the zero state.
}
\label{fig:app:Z8-stabilizer-circuit}
\end{figure}

We now have all the tools we need to simplify Eq.~\eqref{eq:app:generic-generator-with-ancillae}.
We start with
\begin{equation}
\begin{aligned}
\left(U_k^{(\alpha)}\right)^{\dagger}Z_{(\alpha,k)}U_k^{(\alpha)}=
\left(U_k^{(\alpha)}\right)^{\dagger}Z_{(\alpha,k)}\quad
&\Ualpha{k-1}{k}\\
&\Ualpha{k-2}{k-1}\\
&\Ualpha{k-3}{k-2}\\
\vdots\\
&\Ualpha{1}{2}\\
&X_{(A,0),(B,0)}^{(\alpha,1)}.
\label{eq:app:extended-conjugation}
\end{aligned}
\end{equation}
After conjugating the $Z$ operator with the three Toffoli in the first line, we get
\begin{equation}
\begin{aligned}
\left(U_k^{(\alpha)}\right)^{\dagger}Z_{(\alpha,k)}U_k^{(\alpha)}=
\left(U_{k-1}^{(\alpha)}\right)^{\dagger}\quad
&\mathcal{Z}_{\Lambda_2^{(\alpha)}}\\
&\Ualpha{k-2}{k-1}\\
&\Ualpha{k-3}{k-2}\\
\vdots\\
&\Ualpha{1}{2}\\
&X_{(A,0),(B,0)}^{(\alpha,1)}\\
&Z_{(\alpha,k)}\mathcal{Z}_{\Gamma_2^{(\alpha)}},
\end{aligned}
\end{equation}
where $\mathcal{Z}_{\Lambda_2^{(\alpha)}}=Z_{(\alpha,k-1),(A,k-1)}Z_{(\alpha,k-1),(B,k-1)}$ and $\mathcal{Z}_{\Gamma_2^{(\alpha)}}=Z_{(A,k-1),(B,k-1)}$.
Next, we need to conjugate $\mathcal{Z}_{\Lambda_2^{(\alpha)}}$, which does not commute only with the second line,
\begin{equation}
\begin{aligned}
\left(U_k^{(\alpha)}\right)^{\dagger}Z_{(\alpha,k)}U_k^{(\alpha)}=
\left(U_{k-2}^{(\alpha)}\right)^{\dagger}\quad
&\mathcal{Z}_{\Lambda_3^{(\alpha)}}\\
&\Ualpha{k-3}{k-2}\\
\vdots\\
&\Ualpha{1}{2}\\
&X_{(A,0),(B,0)}^{(\alpha,1)}\\
&Z_{(\alpha,k)}\mathcal{Z}_{\Lambda_2^{(\alpha)}}\mathcal{Z}_{\Gamma_2^{(\alpha)}}\mathcal{Z}_{\Gamma_3^{(\alpha)}},
\end{aligned}
\end{equation}
where $\mathcal{Z}_{\Lambda_3^{(\alpha)}}=Z_{(\alpha,k-2),(A,k-2),(A,k-1)}Z_{(\alpha,k-2),(A,k-2),(B,k-1)}Z_{(\alpha,k-2),(B,k-2),(A,k-1)}Z_{(\alpha,k-2),(B,k-2),(B,k-1)}$ and $\mathcal{Z}_{\Gamma_3^{(\alpha)}}=Z_{(A,k-2),(B,k-2),(A,k-1)}Z_{(A,k-2),(B,k-2),(B,k-1)}$.
We need to proceed like this line by line.

We now describe a recursive algorithm to build all the operators that appear.
In the following, we denote a list of qubit labels by $|a,b,c,\dots|$, where the order does not matter.
Start with the single-entry list $\sigma_1^{(1)}=|(\alpha,k)|$ and define the set $\Lambda_1^{(\alpha)}=\{\sigma_1^{(1)}\}$. When conjugating $Z_{(\alpha, k)}$, for each Toffoli operator $T$ inside $U_k^{(\alpha)}$ with target on the $(\alpha,k)$ qubit, take the list obtained by adding the controls of $T$ to $\sigma_1^{(1)}$ and by removing $(\alpha,k)$.
Since there are $3$ Toffoli gates in $U_k^{(\alpha)}$ with target $(\alpha, k)$, this procedure defines three two-entry lists
\begin{equation}
\sigma_1^{(2)}=|(\alpha,k\m1),(A,k\m1)|\qquad
\sigma_2^{(2)}=|(\alpha,k\m1),(B,k\m1)|\qquad
\sigma_3^{(2)}=|(A,k\m1),(B,k\m1)|.
\end{equation}
The first two ones have an entry on the ancillary register, and we put them in the set $\Lambda_2^{(\alpha)}=\{\sigma_1^{(2)},\sigma_2^{(2)}\}$.
The third one has entries only on the links $A$ and $B$, so we put it in the set $\Gamma_2^{(\alpha)}=\{\sigma_3^{(2)}\}$.
This corresponds to the conjugation of $Z_{(\alpha,k)}$ by the first line of Toffoli operators in Eq.~\eqref{eq:app:extended-conjugation}.

Next, for each Toffoli operator $T$ with target $(\alpha,k\m1)$, and for each list $\sigma$ in $\Lambda_2^{(\alpha)}$, take the list obtained by adding the controls of $T$ to $\sigma$ and by removing $(\alpha,k\m1)$.
Since there are three such Toffoli gates and two lists in $\Lambda_2^{(\alpha)}$,
we obtain six three-entry lists
\begin{equation}
\begin{aligned}
\sigma_1^{(3)}&=|(\alpha,k\m2),(A,k\m2),(A,k\m1)|
&
\sigma_2^{(3)}&=|(\alpha,k\m2),(B,k\m2),(A,k\m1)|\\
\sigma_3^{(3)}&=|(\alpha,k\m2),(A,k\m2),(B,k\m1)|
&
\sigma_4^{(3)}&=|(\alpha,k\m2),(B,k\m2),(B,k\m1)|\\
\sigma_5^{(3)}&=|(A,k\m2),(B,k\m2),(A,k\m1)|
&
\sigma_6^{(3)}&=|(A,k\m2),(B,k\m2),(B,k\m1)|.
\end{aligned}
\end{equation}
We put the four lists with an entry in the $\alpha$ register in $\Lambda_3^{(\alpha)}$, and the remaining two in $\Gamma_3^{(\alpha)}$.
This corresponds with the conjugation of $\mathcal{Z}_{\Lambda_2^{(\alpha)}}$.

Given a generic $p=1,\dots,k$, for each Toffoli operator $T$ with target $(\alpha,k-p+1)$, and for each list $\sigma$ in $\Lambda_p^{(\alpha)}$, take the list obtained by removing $(\alpha,k-p+1)$ from $\sigma$ and adding the controls of $T$.
This defines $3n_p$ $(p\p1)$-entry lists, where $3$ is the number of Toffoli operators $T$ in a line, and $n_p$ is the number of lists in $\Lambda_p^{(\alpha)}$.
Among the obtained lists, we put the $2n_p$ ones with an entry in the $\alpha$ register in $\Lambda_{p+1}^{(\alpha)}$, and the remaining $n_p$ ones in $\Gamma_{p+1}^{(\alpha)}$.
We have $n_1=1$, and $n_{p+1}=2n_p=2^p$, while the number of lists in $\Gamma_{p+1}^{(\alpha)}$ is $m_{p+1}^{(\alpha)}=n_p=2^{p-1}$.
The only exception is  $\Lambda_{k+1}^{(\alpha)}$, which is empty due to the form of the last line in Eq.~\eqref{eq:app:extended-conjugation}.

All the lists we have built are contained in the set $\Lambda_{\alpha}\cup\Gamma_{\alpha}$, with $\Lambda_{\alpha}=\Lambda_1^{(\alpha)}\cup\Lambda_2^{(\alpha)}\cup\dots\cup\Lambda_{k+1}^{(\alpha)}$ and $\Gamma_{\alpha}=\Gamma_1^{(\alpha)}\cup\Gamma_2^{(\alpha)}\cup\dots\cup\Gamma_k^{(\alpha)}$.
This allows us to write the result in a compact form,
\begin{equation}
\left(U_k^{(\alpha)}\right)^{\dagger}Z_{(\alpha,k)}U_k^{(\alpha)}=
\prod_{\sigma\in\Lambda_{\alpha}\cup\Gamma_{\alpha}}
Z_{\Sigma_1,\dots,\Sigma_{|\sigma|}}\;,
\end{equation}
where each list $\sigma\in\Lambda_{\alpha}\cup\Gamma_{\alpha}$ contains $|\sigma|$ entries denoted $\Sigma_1,\dots,\Sigma_{|\sigma|}$.
Note that every gate is now diagonal, which means that they all commute with each others and the order does not matter.
As the qubits $(\alpha,l)$, $l=1,\dots,k$ are initially set in the $\ket{0}$ state, we can remove $Z_{(\alpha,k)}$ and all the operators with one control on these qubits.
Thus, on the subspace where the ancillary qubits are in the $\ket{0}$ state, we can write
\begin{equation}
\left(U_k^{(\alpha)}\right)^{\dagger}
Z_{(\alpha,k)}
U_k^{(\alpha)}
=
\prod_{\sigma\in\Lambda_{\alpha}\cup\Gamma_{\alpha}}
Z_{\Sigma_1,\dots,\Sigma_{|\sigma|}}
=
\prod_{\sigma\in\Gamma_{\alpha}}
Z_{\Sigma_1,\dots,\Sigma_{|\sigma|}}.
\label{eq:app:conjugation-simplification}
\end{equation}
The number of operators in the last product is
\begin{equation}
M_{\alpha}=\sum_{p=2}^{k+1}m_p^{(\alpha)}=\sum_{p=2}^{k+1}2^{p-2}=2^k-1.
\end{equation}
The lists in $\Gamma_{p+1}^{(\alpha)}$ with $p>1$ are all the possible lists of the form
\begin{equation}
|(A,k-p),(B,k-p),(X_1,k-p+1),(X_2,k-p+2),\dots,(X_{p-1},k-1)|,
\end{equation}
where $X_1\dots X_{p-1}$ is any possible string made of $A$ and $B$.
This concludes the part involving the $\alpha$ register.

Concerning
\begin{equation}
\left(U_k^{(\beta)}\right)^{\dagger}Z_{(\beta,k)}U_k^{(\beta)},
\end{equation}
we can simply replace $\alpha$ with $\beta$, $A$ with $C$, and $B$ with $D$ everywhere.
The only relevant difference is that, due to the fact that the site $\mathbf{s}$ is involved in the computation of the $\beta$ carry bits, the last line of Eq.~\eqref{eq:app:extended-conjugation} consists of three Toffoli operators instead of one,
\begin{equation}
X_{\mathbf{s},(C,0)}^{(\beta,1)}
X_{\mathbf{s},(D,0)}^{(\beta,1)}
X_{(C,0),(D,0)}^{(\beta,1)}\;.
\end{equation}

As a consequence, $\Gamma_{k+1}^{(\beta)}$ contains now $3\cdot2^{k-1}$ instead of $2^{k-1}$ lists, of the form
\begin{equation}
\begin{gathered}
|(C,0),(D,0),(Y_1,1),(Y_2,2),\dots,(Y_{k-1},k-1)|,\\
|\mathbf{s},(C,0),(Y_1,1),(Y_2,2),\dots,(Y_{k-1},k-1)|,\\
|\mathbf{s},(D,0),(Y_1,1),(Y_2,2),\dots,(Y_{k-1},k-1)|,
\end{gathered}
\end{equation}
where $Y_1\dots Y_{k-1}$ is any possible string made of $C$ and $D$.
The total number of lists in $\Gamma_{\beta}$ is
\begin{equation}
M_{\beta}=\sum_{p=2}^{k+1}m_p^{(\beta)}=\sum_{p=2}^k2^{p-2}+3\cdot2^{k-1}=2^{k+1}-1.
\end{equation}

Putting together Eq.~\eqref{eq:app:generic-generator-with-ancillae}, Eq.~\eqref{eq:app:conjugation-simplification}, and the corresponding result for $\beta$ we find the final form of the $k$\textsuperscript{th} generator,
\begin{equation}
\begin{gathered}
g_{\mathbf{s},k}=
Z_{(A,k)}\otimes
Z_{(B,k)}\otimes
Z_{(C,k)}\otimes
Z_{(D,k)}\otimes
\left(
\prod_{\sigma\in\Gamma_{\alpha}}
Z_{\Sigma_1,\dots,\Sigma_{|\sigma|}}
\right)
\otimes
\left(
\prod_{\sigma\in\Gamma_{\beta}}
Z_{\Sigma_1,\dots,\Sigma_{|\sigma|}}
\right),\\
\Gamma_{\alpha}=
\Gamma_2^{(\alpha)}\cup\dots
\cup\Gamma_{k+1}^{(\alpha)},
\qquad
\Gamma_{\beta}=\Gamma_2^{(\beta)}\cup\dots\cup\Gamma_{k+1}^{(\beta)},
\end{gathered}
\end{equation}
where $\Sigma_1,\dots,\Sigma_{|\sigma|}$ are the $\abs{\sigma}$ entries of the list $\sigma$, and
\begin{itemize}
\item $\Gamma_2^{(\alpha)}$ contains the single list $|(A,k\m1),(B,k\m1)|$;
\item $\Gamma_p^{(\alpha)}$, $p=3,\dots,k+1$, contains all the possible $p$-entry lists of indexes
\begin{equation}
|(A,k-p+1),(B,k-p+1),(X_1,k-p+2),(X_2,k-p+3),\dots,(X_{p-2},k-1)|,
\end{equation}
with $X_1\dots X_{p-2}$ an arbitrary string of $A$ and $B$;
\item $\Gamma_2^{(\beta)}$ contains the single list $|(C,k\m1),(D,k\m1)|$;
\item $\Gamma_p^{(\beta)}$, $p=3,\dots,k$, contains all the possible $p$-entry lists of indexes
\begin{equation}
|(C,k-p+1),(D,k-p+1),(Y_1,k-p+2),(Y_2,k-p+3),\dots,(Y_{p-2},k-1)|,
\end{equation}
with $Y_1\dots Y_{p-2}$ an arbitrary string of $C$ and $D$;
\item $\Gamma_{k+1}^{(\beta)}$ contains all the possible $(k\p1)$-entry lists of indexes
\begin{equation}
\begin{gathered}
|(C,0),(D,0),(Y_1,1),(Y_2,2),\dots,(Y_{k-1},k-1)|,\\
|\mathbf{s},(C,0),(Y_1,1),(Y_2,2),\dots,(Y_{k-1},k-1)|,\\
|\mathbf{s},(D,0),(Y_1,1),(Y_2,2),\dots,(Y_{k-1},k-1)|,
\end{gathered}
\end{equation}
with $Y_1\dots Y_{p-2}$ an arbitrary string of $C$ and $D$.
\end{itemize}
All the entries in the lists denote qubits on the lattice up to level $k\m1$.
Defining
\begin{equation}
\mathcal{C}_{\mathbf{s},k}=
\left(
\prod_{\sigma\in\Gamma_{\alpha}}
Z_{\Sigma_1,\dots,\Sigma_{|\sigma|}}
\right)\otimes
\left(
\prod_{\sigma\in\Gamma_{\beta}}
Z_{\Sigma_1,\dots,\Sigma_{|\sigma|}}
\right),
\end{equation}
we can write the generators in the form used in the main text
\begin{equation}
g_{\mathbf{s},k}=
Z_{(A,k)}\otimes
Z_{(B,k)}\otimes
Z_{(C,k)}\otimes
Z_{(D,k)}\otimes
\mathcal{C}_{\mathbf{s},k}.
\end{equation}

For a site $\mathbf{s}$ with odd $\abs{\mathbf{s}}$, we have to modify the result in two ways looking at Eq.~\eqref{eq:2d:even-conservation-law} and Eq.~\eqref{eq:2d:odd-conservation-law}.
First, we need to replace $A$ with $C$, $B$ with $D$ and vice versa.
Furthermore, since $f_{\mathbf{s}}$ is substituted by $1-f_{\mathbf{s}}$, we need to replace the controls on the site $\mathbf{s}$ with anticontrols.
Notice however that the association between $A$, $B$, $C$ and $D$ and east, north, west and south remains unchanged.

\bibliography{Biblio_QEC}

\end{document}